\documentclass[
aps,
prxquantum,
twocolumn,
notitlepage,
superscriptaddress,
longbibliography
]{revtex4-2}

\usepackage[T1]{fontenc}
\usepackage{amsmath,amssymb,amsthm,amsfonts}
\usepackage{bm}
\usepackage{mathtools}
\usepackage{mathrsfs}
\usepackage{graphicx}
\usepackage{xcolor}
\usepackage{hyperref}

\newcommand{\rank}{\operatorname{rank}}

\theoremstyle{plain}
\newtheorem{theorem}{Theorem}[section]

\newtheorem{lemma}[theorem]{Lemma}
\newtheorem{proposition}[theorem]{Proposition}
\theoremstyle{definition}

\theoremstyle{remark}

\begin{document}

\title{Finite Realizations and Effective Memory in Nonlinear Quantum Filtering}

\author{Jacob Emerson}
\affiliation{Department of Electrical and Computer Engineering, Princeton University}

\date{\today}

\begin{abstract}
Nonlinear dynamics generate non-Gaussian states and operations, but they also make continuously monitored quantum systems difficult to track. A quantum filter performs this task by compressing a noisy measurement record into a set of evolving variables that predict future observables. For Gaussian dynamics, this compression is exact and finite. In nonlinear systems, conventional moment equations generally expand without bound, but the failure of one representation does not prove that every finite description is impossible.
Here we formulate quantum filtering as a realization problem for general quantum input--output maps. We introduce the linearized observable Hankel operator \(K_O\), which maps changes in the past measurement record to changes in a future conditional observable. Any finite-dimensional smooth filter forces \(K_O\) to have finite rank, so infinite rank provides a coordinate-independent proof that no exact finite-dimensional filter exists. Applied to single-mode polynomial bosonic systems under continuous quadrature monitoring, this yields a sharp boundary: Gaussian dynamics and the measurement-aligned Conditional Momentum Moment class admit finite realizations, whereas dynamics outside these classes produce infinitely many response directions and admit no finite-dimensional smooth exact filter.
Exact impossibility, however, need not imply large practical complexity. The singular values of \(K_O\) quantify how many response directions matter at a chosen local accuracy, defining the effective memory of an observable. For the Kerr and Duffing oscillators, these singular values decay rapidly despite infinite exact rank, showing that infinite exact memory can coexist with low finite-accuracy complexity. The framework therefore separates exact realizability from effective complexity in nonlinear quantum filtering.
\end{abstract}

\maketitle

\section{Introduction}
\label{sec:introduction}

Continuous quantum measurement has become a standard method of readout in quantum optical and superconducting systems. By monitoring part of the field emitted into the environment, homodyne or heterodyne detection produces a time-dependent record of an evolving quantum system \cite{belavkin1992,wiseman_milburn2009,jacobs2014}. Such records are now used to reconstruct quantum trajectories, stabilize states through feedback, and even anticipate and reverse individual quantum jumps \cite{murch2013observing,vijay2012feedback,weber2014optimal,minev2019jump}. Continuous measurement thus provides information throughout an experiment, while its noise and backaction also drive the dynamics conditioned on that information.

Averaging over measurement records recovers the deterministic Lindblad evolution, which is sufficient for many purposes, but loses information about individual runs. Keeping the record allows observables to be estimated conditionally and used for feedback. In mechanical systems, such estimates have enabled feedback control and helped distinguish the variance of a conditional state from the broader variation seen across runs \cite{rossi2018control,rossi2019trajectory}. In a different setting, the intrinsic dimension of quantum trajectory data has been shown to reflect dynamical structure, including integrability, Hilbert-space fragmentation, and closure of correlation hierarchies \cite{lumia2026complexity}.

Quantum reservoir computing offers another reason to study dynamics under measurement. In these proposals, a time-dependent input drives a quantum system, and measured outputs are used to process the input history. Nonlinear bosonic reservoirs and protocols using weak measurements have been proposed for temporal tasks \cite{govia2021reservoir,khan2021physical,mujal2023timeseries}. Measurement backaction can itself contribute to the reservoir's processing and memory \cite{franceschetto2026backaction}. These results make the relation between measurement, dynamics, and retained temporal information a concrete question.

A quantum filter uses the measurement record up to the present time to estimate an observable, a collection of correlations, or the conditional quantum state \cite{belavkin1992,wiseman_milburn2009,jacobs2014}. Monitoring one field quadrature, such as \(X\), can thereby provide an estimate of a quantity that is not measured directly. For feedback, recentering, or a trained readout, the desired observables may suffice without reconstructing every component of the conditional state \cite{khan2024postprocessor}. A finite-dimensional filter would produce these estimates by evolving a fixed number of internal variables as the record grows.

For Gaussian states under Gaussian-preserving dynamics, the conditional state is determined by its first and second moments \cite{weedbrook2012gaussian}. Nonlinear Hamiltonians generally destroy this closure, but they also generate physics beyond the Gaussian manifold, including Wigner negativity, Schrödinger-cat superpositions, collapse and revival, and resources for universal continuous-variable quantum computation \cite{kudra2022wigner,kirchmair2013collapse,lloydbraunstein1999}. Under continuous monitoring, the evolution of an observable generally couples it to additional correlations. For example, \(X\)-homodyne monitoring can schematically generate
\[
\langle O\rangle_c
\longrightarrow
\langle OX\rangle_c
\longrightarrow
\langle OX^2\rangle_c
\longrightarrow\cdots ,
\]
while the Hamiltonian introduces further couplings. In a direct moment description, the number of required variables appears to grow without bound.

An unbounded moment hierarchy does not, by itself, rule out a finite filter: a different choice of variables might close. The Conditional Momentum Moment (CMM) construction demonstrates this possibility for a class of measurement-aligned nonlinear Hamiltonians \cite{emersonCMM}. Other finite structures arise from linear Heisenberg dynamics, invariant families of conditional states, and estimation algebras \cite{warszawskiwisemandoherty2020,sarletteelouardrouchon2025,mitter1979,brockett1981,chiouyau1994,aminigough2019,grigolettoPellegriniTicozzi2025}. These examples leave a general question: given a conditional observable, can any finite-dimensional causal filter reproduce its dependence on the measurement record?

This question is the quantum version of a longstanding problem in nonlinear filtering. Outside a small number of structured cases, nonlinear filters have generally been expected to require an infinite-dimensional description, and the search for exact finite-dimensional exceptions has motivated classification programs since the early development of nonlinear filtering theory \cite{mitter1979,brockett1981,marcus1984,chiouyau1994}. What is harder to establish is that this apparent complexity is intrinsic: failure of a particular hierarchy, parametrization, or algebraic construction need not exclude some different finite realization. Here we close this loophole for one-mode polynomial bosonic dynamics, giving an exact boundary between the systems for which finite realization is possible and those for which no smooth choice of internal variables can make the exact observable filter finite dimensional.

We approach this question as an observable-level realization problem where the input is the past measurement record, and the output is a future conditional observable. Perturbing the record and calculating the resulting first-order change in the output defines the linearized observable Hankel operator \(K_O\). If a smooth filter has \(d\) internal variables, its response to every past perturbation must pass through those variables, so \(K_O\) has rank at most \(d\). Infinite rank therefore rules out a finite smooth realization of that observable, regardless of the variables used to construct the filter. It also rules out a finite realization of the full conditional state.

The rank test adapts the past--future perspective of classical input--output realization theory to differentiable quantum input--output maps \cite{fliess1974,sussmann1976,arbibmanes1980,grayScherpen}. For the polynomial bosonic setting considered here, the resulting boundary can be stated precisely: Gaussian dynamics and the measurement-aligned CMM class admit finite observable realizations, whereas outside their union every nonconstant polynomial observable generically has infinite Hankel rank and therefore admits no finite-dimensional robust \(C^1\) exact filter.

The exact boundary leaves a practical question: how many variables are needed when some approximation error is acceptable? The singular values of \(K_O\) give a lower bound on the error of a fixed-dimensional approximation to the local observable response \cite{scherpen1993,fujimotoscherpen2010,eckartyoung1936}. They provide a benchmark for the Kerr and Duffing comparisons that follow.

The paper is organized as follows. Sections~\ref{sec:background}--\ref{sec:realization} introduce the monitored bosonic model, its finite closures, and the observable realization framework. Section~\ref{sec:classification} establishes the boundary of exact realizability. Sections~\ref{sec:approximation}--\ref{sec:numerics} develop finite-accuracy approximation and present the Kerr and Duffing results. Technical proofs and detailed constructions appear in the Supplemental Material \cite{supplemental}.

\section{Monitored bosonic dynamics and finite closures}
\label{sec:background}

We begin with the concrete setting: a single bosonic mode with polynomial Hamiltonian dynamics under continuous quadrature monitoring. After fixing the physical model, we examine Gaussian and cubic-phase dynamics as representative examples of the two finite sectors of the classification. Their observable response kernels make the basic rank argument explicit before we formulate it for general quantum input--output maps.

\subsection{Monitored bosonic dynamics}
\label{sec:model}
The bosonic system considered here has mode quadratures satisfying
\begin{equation}
 [X,P]=i,
 \qquad
 a=\frac{X+iP}{\sqrt2}.
\end{equation}
The input record is the homodyne or heterodyne measurement current, and the outputs are
conditional expectations of system observables.  Since all expectations and states below
are conditioned on the observed record, we suppress an additional
conditional subscript.

Throughout the main classification, the Hamiltonian is the Weyl ordering of a
finite real polynomial,
\begin{equation}
 h(x,p)=\sum_{\ell,k}h_{\ell k}x^\ell p^k,
 \qquad
 H=[h]_{\mathrm W}.
 \label{eq:main-weyl-H}
\end{equation}
For monomials, Weyl ordering is symmetrization:
\([x^\ell p^k]_{\mathrm W}=\operatorname{Sym}(X^\ell P^k)\), where
\(\operatorname{Sym}(X^\ell P^k)\) is the average over all distinct orderings
of \(\ell\) copies of \(X\) and \(k\) copies of \(P\).  Thus real
phase-space polynomials are represented by Hermitian polynomial operators.
This convention also pairs naturally with the Wigner representation, where Hamiltonians and observables are tracked through their ordinary
phase-space symbols \cite{folland1989}.

Damping into a thermal Markov bath of occupation \(\bar n\), together with
homodyne monitoring of the loss channel, gives the stochastic master equation
\begin{align}
 d\rho_t
 &={\cal L}_{\bar n}\rho_t\,dt
   +\sqrt{\eta\kappa}\,{\cal H}[a]\rho_t\,dW_t,
 \label{eq:main-SME}\\
 dY_t
 &=\sqrt{2\eta\kappa}\,\langle X\rangle_t\,dt+dW_t,
 \qquad
 \langle O\rangle_t=\operatorname{Tr}(O\rho_t),
 \label{eq:main-current}
\end{align}
where
\begin{align}
 {\cal L}_{\bar n}\rho
 &=-i[H,\rho]
 +\kappa(\bar n+1){\cal D}[a]\rho
 +\kappa\bar n{\cal D}[a^\dagger]\rho,\\
 {\cal D}[c]\rho
 &=c\rho c^\dagger-\tfrac12\{c^\dagger c,\rho\},\\
 {\cal H}[a]\rho
 &=a\rho+\rho a^\dagger
 -\operatorname{Tr}[(a+a^\dagger)\rho]\rho.
\end{align}
Here \(\kappa>0\) is the damping rate, \(0<\eta\le1\) is the detection
efficiency, and \(\bar n\ge0\) is the bath occupation. 

The realization argument concerns the map from the observed record to the
observable trajectory.  For this purpose, it is useful to remove the
normalization nonlinearity from the evolution and restore it only at the
readout.  Writing
\[
 \rho_t=\frac{\widetilde\rho_t}{\operatorname{Tr}\widetilde\rho_t},
\]
the unnormalized conditional state obeys the linear equation
\begin{equation}
 d\widetilde\rho_t
 ={\cal L}_{\bar n}\widetilde\rho_t\,dt
 +{\cal M}_\eta\widetilde\rho_t\,dY_t,
 \qquad
 {\cal M}_\eta\rho
 =\sqrt{\eta\kappa}(a\rho+\rho a^\dagger).
 \label{eq:main-zakai}
\end{equation}
The physical conditional expectation is recovered by normalizing the output
\begin{equation}
 \langle O\rangle_t
 =
 \frac{\operatorname{Tr}(O\widetilde\rho_t)}
 {\operatorname{Tr}(\widetilde\rho_t)}.
 \label{eq:main-normalized-output}
\end{equation}
Thus the record-driven evolution is linear before the final scalar
normalization.

We will also pass to the Wigner representation, the quantum analogue of classical
phase space. Here \(x\) and \(p\) are ordinary
phase-space variables corresponding to the quadratures \(X\) and \(P\).  If
\(W_t(x,p)\) is the Wigner transform of the unnormalized state
\(\widetilde\rho_t\), then
\begin{equation}
 dW_t=({\cal H}_h+{\cal D}_{\bar n})W_t\,dt
       +{\cal M}_\eta W_t\,dY_t,
 \label{eq:main-wigner-zakai}
\end{equation}
with
\begin{align}
 {\cal D}_{\bar n}W
 &=\frac\kappa2[\partial_x(xW)+\partial_p(pW)]
 +\frac{\kappa(2\bar n+1)}4
  (\partial_x^2+\partial_p^2)W,\\
 {\cal M}_\eta W
 &=\sqrt{2\eta\kappa}
   \left(x+\tfrac12\partial_x\right)W.
 \label{eq:main-measurement-letter}
\end{align}
The Hamiltonian term \({\cal H}_h\) is the exact Moyal commutator generated by
\(H=[h]_{\mathrm W}\).  Its odd-derivative expansion terminates because \(h\)
is a polynomial.

\subsection{The finite classes: Gaussian and CMM}
\label{sec:finite-classes-background}

We introduce the finite structures that form the solvable side of the
classification. We state the one-mode \(X\)-homodyne class to fix notation,
while Fig.~\ref{fig:main-hamiltonian-classes} summarizes the
boundary for \(P\)-homodyne monitoring as well.  The rotated and multimode versions are discussed with
the main theorem and in the Supplemental Material \cite{supplemental}.

The Gaussian class consists of Hamiltonians at most quadratic in the
quadratures,
\begin{equation}
 H\in{\mathscr G}
 \qquad\Longleftrightarrow\qquad
 H=[h]_{\mathrm W},\quad \deg h\le 2.
 \label{eq:main-gaussian-class}
\end{equation}
For these systems, continuous Gaussian measurements preserve Gaussian
conditional states, so the filter closes on finitely many means and
covariances.  This includes the usual homodyne and heterodyne Gaussian
filters and related finite-dimensional closure constructions
\cite{warszawskiwisemandoherty2020}.

The second finite class is the Conditional Momentum Moment class introduced
in Ref.~\cite{emersonCMM}.  For homodyne monitoring of \(X\), it consists of
Hamiltonians of the form
\begin{equation}
 H\in{\mathscr C}_X
 \qquad\Longleftrightarrow\qquad
 H=V(X)+aP+b\,\operatorname{Sym}(XP),
 \label{eq:main-CMM-class}
\end{equation}
with \(V\) a polynomial and \(a,b\in\mathbb R\).  The useful coordinates are
not the full Wigner function, but the conditional moments of the conjugate
quadrature at fixed measured quadrature,
\begin{equation}
 m_n(x,t)
 =
 \frac{\int_{\mathbb R} p^n W_t(x,p)\,dp}
      {\int_{\mathbb R} W_t(x,p)\,dp}.
 \label{eq:main-CMM-moments}
\end{equation}
For Hamiltonians in \({\mathscr C}_X\), these quantities obey a triangular
system: the evolution of \(m_n\) involves only lower momentum moments and
polynomial functions of \(x\).  For any fixed observable degree, the required
moments therefore close on finitely many polynomial coefficients, from which
mixed Weyl-ordered observables can also be recovered.

The CMM class does not fully contain the Gaussian class.  For example,
\(X^2+P^2\) is Gaussian but does not belong to \({\mathscr C}_X\), because it
contains quadratic dependence on the conjugate quadrature.  Instead, CMM
extends the finite sector in a different direction by allowing arbitrary
polynomial dependence on the measured quadrature, including nonlinearities like cubic-phase
dynamics \(H\propto X^3\).  Its existence is therefore structurally important:
it gives a genuinely nonlinear exact closure and shows why the failure of a
standard moment hierarchy cannot by itself rule out a finite filter.

The closure is also constructive rather than purely formal.  In
Ref.~\cite{emersonCMM}, it was used to obtain a causal real-time estimator for
the nonlinear nullifier of a squeezed cubic-phase state and to simulate
measurement-based feedback recentering without propagating the full
conditional state.  Together, these two closures form the finite side of the one-mode
\(X\)-homodyne classification,
\begin{equation}
 {\mathscr F}_X={\mathscr G}\cup{\mathscr C}_X.
 \label{eq:main-finite-class}
\end{equation}

\section{Observable realization theory}
\label{sec:realization}

We now introduce the coordinate-independent realization framework used throughout the paper. This formulation allows us to import tools from classical input--output theory and test for the existence of finite-dimensional quantum filters without specifying their internal coordinates.

\subsection{Finite realizations and the observable Hankel operator}
\label{sec:KO}

A causal filter converts an observed record into the trajectory of the
quantities one wants to predict.  A \(d\)-dimensional finite-dimensional filter represents all
record information needed for those outputs by internal coordinates
\begin{equation}
 dz_t=f(z_t,t)\,dt+g(z_t,t)\,dY_t,
 \qquad
 \widehat O_t=q(z_t,t),
 \label{eq:abstract-filter}
\end{equation}
where \(z_t\in\mathbb R^d\). Causality means that \(z_t\), and hence
\(\widehat O_t\), may depend only on the record up to time \(t\).
We consider \(C^1\) realizations that remain exact on an open neighborhood
of input records, so that their response to record perturbations is
differentiable. The coordinates may be moments, Fock coefficients,
Gaussian parameters, or nonlinear variables with no prescribed physical
meaning.

We now define the linearized input--output operator used throughout the
classification.  Fix a final time \(T\), a splitting time \(\tau\), and write
\begin{equation}
 I_-=[0,\tau],
 \qquad
 I_+=[\tau,T].
 \label{eq:main-past-future-intervals}
\end{equation}
The interval \(I_-\) represents the past record, while \(I_+\) represents the
future output window.  For a chosen reference input \(u\), the past-to-future
observable map is
\begin{equation}
 {\mathscr H}_O:
 u|_{I_-}
 \longmapsto
 \left.
 \langle O\rangle_t
 \right|_{I_+}.
 \label{eq:main-past-future-map}
\end{equation}
We study this map locally around a reference input \(u\).  Because its input
and output are functions of time, the relevant derivative is the
Fr\'echet derivative, the infinite-dimensional analogue of a Jacobian:
it is the linear map that gives the first-order change in the future
observable trajectory caused by a small perturbation of the past input.  Thus,
for a perturbation \(v\) supported in \(I_-\),
\begin{equation}
 {\mathscr H}_O(u+\epsilon v)(t)
 =
 {\mathscr H}_O(u)(t)
 +\epsilon\,(K_O v)(t)
 +o(\epsilon),
 \qquad t\in I_+.
 \label{eq:main-frechet-definition}
\end{equation}
The linear operator
\begin{equation}
 K_O=D{\mathscr H}_O(u):L^2(I_-)\longrightarrow L^2(I_+)
 \label{eq:main-KO}
\end{equation}
is the observable Hankel operator. For clarity, we also call \(K_O\) the
linearized response operator, or simply the response operator when the
meaning is clear.

\begin{theorem}[Observable realization obstruction]
\label{thm:general-rank-obstruction}
If a differentiable quantum input--output map for an observable \(O\) admits a
robust \(d\)-dimensional \(C^1\) filter, then
\begin{equation}
 \rank K_O\le d.
 \label{eq:main-rank-bound}
\end{equation}
\end{theorem}

To see the factorization, let \({\mathscr R}_\tau\) map the past record to the
filter state at the splitting time and let \({\mathscr G}_O\) map that state to
the future observable.  Then
\begin{equation}
 {\mathscr H}_O={\mathscr G}_O\circ{\mathscr R}_\tau,
 \qquad
 K_O=D{\mathscr G}_O({\mathscr R}_\tau(u))
 \circ D{\mathscr R}_\tau(u).
 \label{eq:main-filter-factorization}
\end{equation}
Because \({\mathscr R}_\tau\) maps into a \(d\)-dimensional filter state,
the composition defining \(K_O\) can contain at most \(d\) independent
response directions.  This proves Eq.~\eqref{eq:main-rank-bound} and makes the
bound independent of the coordinates used for the filter.

The one-sided nature of this criterion is well suited to the classification
problem considered here. On the finite side, the Gaussian and CMM cases
already provide explicit realizations; for the remaining cases, it is enough
to find an obstruction. Thus, once \(K_O\) is shown to have infinite rank
outside those classes, the constructive and obstructive results together give
the full boundary.

Passing from the nonlinear map \({\mathscr H}_O\) to its linearization is a
deliberate simplification. Directly characterizing how every possible record
affects the future observable would generally be intractable, whereas any
smooth finite realization must also give a finite-dimensional realization of
its linearized response. A single infinite-rank linearization can therefore
rule out the full nonlinear filter. Even \(K_O\) need not be computed
completely: it is enough to identify arbitrarily many independent response
directions. This observable-level test is also stronger than a state-level
calculation in one useful respect: infinite rank for one scalar observable
excludes an exact finite-dimensional state filter, while a chosen observable
may still close on finitely many coordinates when the full conditional state
does not. The CMM class provides precisely such an observable-level closure.

\subsection{The bosonic response kernel}
\label{sec:bosonic-response-kernel}

For the monitored bosonic model, we replace the record increment \(dY_t\) by
a smooth input \(u(t)\,dt\).  In this controlled, pathwise form the Wigner
function evolves by
\begin{equation}
 \begin{aligned}
 \partial_tW_u&=A_u(t)W_u,\\
 A_u(t)&={\cal H}_h+{\cal D}_{\bar n}
 -\tfrac12{\cal M}_\eta^2
 +u(t){\cal M}_\eta.
 \end{aligned}
 \label{eq:main-controlled-zakai}
\end{equation}
The term \(-{\cal M}_\eta^2/2\) is the correction needed when passing from the
It\^o equation to the smooth-record response equation.  This converts the
record dependence into an ordinary linear evolution in time and permits
differentiation with respect to the input.  The passage from this controlled
calculation to physical measurement records is justified in the Supplemental
Material \cite{supplemental}.

Let \(\Phi_u(t,s)\) denote the propagator generated by
Eq.~\eqref{eq:main-controlled-zakai}; thus
\[
 W_u(t)=\Phi_u(t,s)W_u(s),
 \qquad s\le t.
\]
We use the phase-space pairing notation
\begin{equation}
 \langle f,W\rangle
 :=
 \int_{\mathbb R^2} f(x,p)W(x,p)\,dx\,dp.
 \label{eq:main-phase-space-pairing}
\end{equation}
With our Wigner normalization, the conditional expectation of an observable
\(O\) is
\begin{equation}
 \langle O\rangle_t
 =
 \frac{\langle o,W_u(t)\rangle}{\langle 1,W_u(t)\rangle},
 \label{eq:main-observable-pairing}
\end{equation}
where \(o(x,p)\) is the phase-space representative of \(O\).

Duhamel's formula \cite{pazy1983} gives the corresponding first variation of the
unnormalized Wigner function.  A perturbation inserted at time \(s\) produces
\begin{equation}
 \frac{\delta W_u(t)}{\delta u(s)}
 =
 \Phi_u(t,s){\cal M}_\eta W_u(s),
 \qquad s<t.
 \label{eq:main-state-response}
\end{equation}
After applying the normalized readout
\eqref{eq:main-observable-pairing}, the Hankel operator is the integral
operator
\begin{equation}
 (K_O v)(t)
 =
 \int_{I_-}K_O(t,s)v(s)\,ds,
 \qquad t\in I_+,
 \label{eq:main-hankel-integral}
\end{equation}
with kernel
\begin{equation}
 \begin{aligned}
 K_O(t,s)
 &=\frac{\left\langle
 o-\langle O\rangle_t,\,
 \Phi_u(t,s){\cal M}_\eta W_u(s)
 \right\rangle}{\langle1,W_u(t)\rangle}.
 \end{aligned}
 \label{eq:main-memory-kernel}
\end{equation}
Here $s\in I_-$ and $t\in I_+$.
The centered factor \(o-\langle O\rangle_t\) comes from differentiating the
normalization in Eq.~\eqref{eq:main-observable-pairing}.

The operator \(K_O\) measures the first-order influence of the past record on
the future observable trajectory.  Its input is a small perturbation of the
past measurement record, and its output is the induced change in the future
conditional expectation.

\subsection{A finite nonlinear calibration}
\label{sec:finite-side-calibration}

The same calculation shows how aligned CMM dynamics stop instead of producing
an increasing response chain.  Consider the setting and desired observable
\begin{equation}
 H=\gamma X^3,
 \qquad O=P.
\end{equation}
In ordinary moment coordinates, this system does not appear finite.
The stochastic \(X\)-measurement term generates the schematic hierarchy
\[
\langle P\rangle_t
\longrightarrow
\langle XP\rangle_t
\longrightarrow
\langle X^2P\rangle_t
\longrightarrow\cdots ,
\]
so increasingly many moments enter the evolution of
\(\langle P\rangle_t\). This unbounded hierarchy suggests that an exact
filter would require infinitely many variables. The CMM transformation does
not truncate the hierarchy. Instead, it packages all of these
\(X\)-weighted moments into the conditional momentum moment
\[
m_1(x,t)
=
\frac{\int p\,W_c(x,p,t)\,dp}
     {\int W_c(x,p,t)\,dp},
\]
and then treats its dependence on the measured coordinate \(x\) as a single
evolving function.

For the cubic-phase Hamiltonian, the Hamiltonian force is proportional to
\(x^2\partial_p\). It increases the polynomial degree in \(x\), but the
absence of a kinetic term prevents this dependence from being transferred
into higher momentum degree. Writing the result in the centered coordinate
\(\xi=x-\langle X\rangle_t\), the first conditional momentum moment therefore
closes as the quadratic polynomial
\begin{equation}
 \begin{aligned}
 m_1(\xi,t)&=a_0(t)+a_1(t)\xi+a_2(t)\xi^2,\\
 \langle P\rangle_t&=a_0(t)+\frac12a_2(t),
 \end{aligned}
 \label{eq:main-cubic-CMM-polynomial}
\end{equation}
where the second line follows by averaging \(m_1\) over the Gaussian
\(X\)-marginal. For the standard coherent-width measured
marginal, the three coefficient equations are
\begin{align}
 da_2
 &=\left(-3\gamma-\frac{3\kappa}{2}a_2\right)dt,
 \label{eq:main-cmm-a2}\\
 da_1
 &=-\kappa a_1\,dt
 +\sqrt{2\eta\kappa}\,a_2\,dW_t,
 \label{eq:main-cmm-a1}\\
 da_0
 &=\left(-\frac{\kappa}{2}a_0+\frac{\kappa}{2}a_2\right)dt
 +\sqrt{\frac{\eta\kappa}{2}}\,a_1\,dW_t.
 \label{eq:main-cmm-a0}
\end{align}
The first equation is deterministic, and $dW_t$ and $dY_t$ are related by a
causal invertible change of input coordinates.  Only $a_1$ and $a_0$ carry
variations of the measurement record, so the observable response can have at
most two independent record directions even though the conditional state need
not be Gaussian or finite dimensional.

The two directions can be seen directly in the kernel.  Replace $dW_t$ by a
smooth input $u(t)dt$, linearize around a constant $u_0$, and write
$\Delta=t-s$.  Solving the two record-dependent linear equations gives
\begin{equation}
 \begin{aligned}
 K_P(t,s)
 &=\sqrt{\frac{\eta\kappa}{2}}\,
 a_1(s)e^{-\kappa\Delta/2}\\
 &\quad+2\eta u_0a_2(s)
 \left(e^{-\kappa\Delta/2}-e^{-\kappa\Delta}\right).
 \end{aligned}
 \label{eq:main-cmm-explicit-kernel}
\end{equation}
After separating the \(t\)- and \(s\)-dependence, the only future factors are
\(e^{-\kappa t/2}\) and \(e^{-\kappa t}\), so the kernel has rank at most
two. For \(\gamma u_0\ne0\), the corresponding past factors are generically
independent, and the rank is exactly two. This agrees with the structure of
the CMM closure: \(a_2\) evolves deterministically, while all dependence on
the measurement record passes through \(a_1\) and \(a_0\). These two
record-dependent coordinates are therefore sufficient, and the rank of
\(K_P\) shows that they are also locally necessary. Although the cubic-phase
Hamiltonian is nonlinear and generates an infinite hierarchy of ordinary
moments, its lack of quadrature mixing prevents these two response directions
from developing into an unbounded chain. This example also shows how the internal response directions of the actual
filter appear directly in \(K_P\), allowing the Hankel factorization to
provide structural guidance as well as an obstruction.

\section{The boundary of exact realizability}
\label{sec:classification}
With the finite classes defined, we can state the exact boundary.  The result
says that outside
${\mathscr F}_X={\mathscr G}\cup{\mathscr C}_X$, even a single polynomial
observable has infinite past-to-future response rank.

\begin{figure}[t]
\centering
\includegraphics[width=\columnwidth]{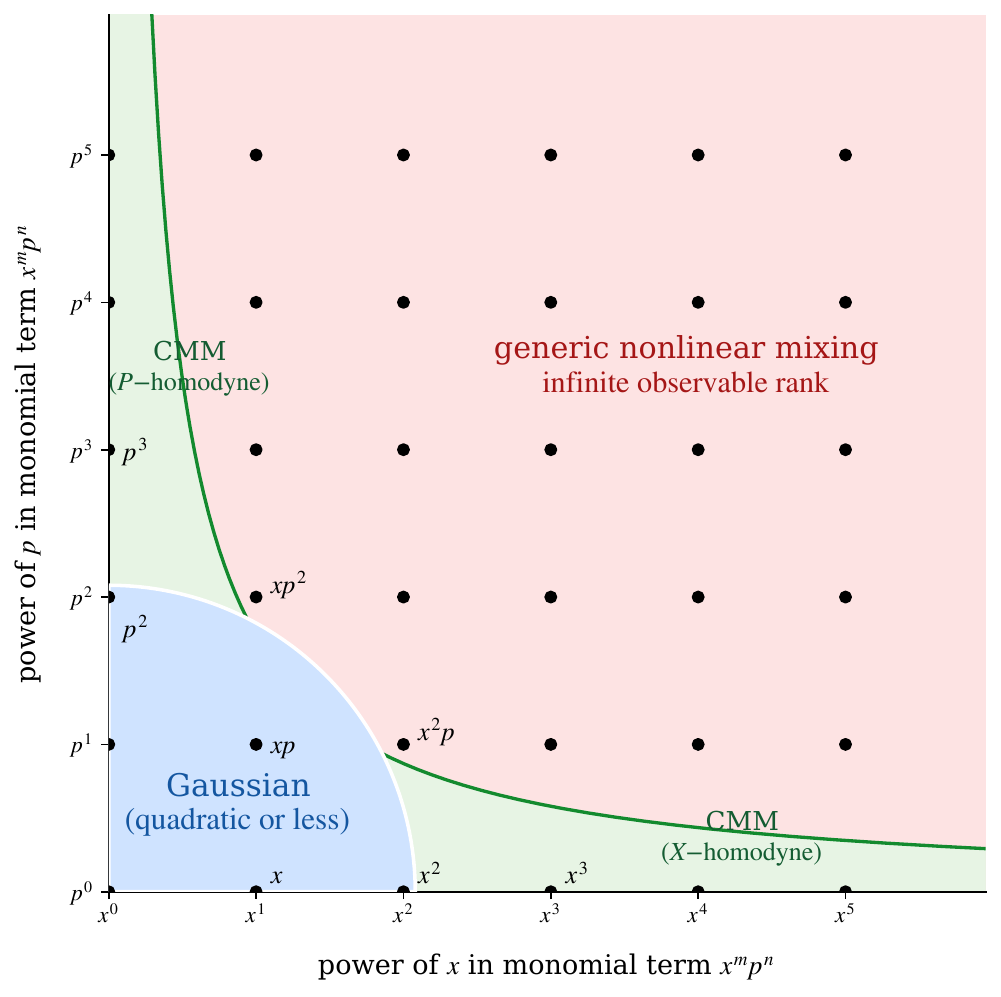}
\caption{
Single-mode schematic of the finite-rank and infinite-rank Hamiltonian sectors.
Each lattice point \((\ell,k)\) represents a Weyl-ordered monomial
\(\operatorname{Sym}(X^\ell P^k)\) appearing in the Hamiltonian.  The finite
sectors shown are the Gaussian class and the CMM classes aligned with
\(X\)- and \(P\)-homodyne monitoring.  Monomials outside these sectors mix the
measured and conjugate quadratures in a way that produces infinite observable
Hankel rank in the classification theorem.
}
\label{fig:main-hamiltonian-classes}
\end{figure}

\begin{theorem}[One-mode polynomial classification]
\label{thm:main-classification}
Fix $\kappa>0$, $0<\eta\le1$, $\bar n\ge0$, and $X$-homodyne
monitoring.  Let $H=[h]_{\mathrm W}$ be a time-independent finite real
Weyl-polynomial Hamiltonian with $H\notin{\mathscr F}_X$.  For every
nonconstant Weyl monomial observable
\begin{equation}
 O_{m,n}=\operatorname{Sym}(X^mP^n),
 \qquad(m,n)\ne(0,0),
\end{equation}
there is a constant controlled reference input such that, for almost every
Gaussian initial preparation,
\begin{equation}
 \boxed{\rank K_{O_{m,n}}=\infty.}
\end{equation}
Consequently no finite-dimensional robust $C^1$ filter can reproduce the
exact record-to-observable map for $\langle O_{m,n}\rangle_t$ in a
neighborhood of that reference input.
\end{theorem}

For the basic nonexistence result, several of these qualifiers are not required.
A global \(d\)-dimensional \(C^1\) filter would imply
\(\rank K_O\leq d\) at every admissible reference record, so a single
infinite-rank response is enough to produce a contradiction. We prove the
stronger statement that this obstruction persists on open neighborhoods of
records and for almost every Gaussian preparation. This shows that the result
does not depend on a finely tuned record or initial state and rules out finite
realizability arising from fragile, isolated cancellations.

Here ``almost every'' is taken with respect to the usual displacement and
covariance parameters, and separately with respect to displacement and
squeezing on the pure-Gaussian family. For each size \(N\), the selected
response minor is an analytic function of these parameters and is nonzero for
a convenient coherent preparation. Its zero set therefore has measure zero,
and a countable intersection gives a full-measure set on which minors of every
size remain nonzero. The excluded measure-zero preparations are only points
at which the particular minors used in the proof may cancel or change their
leading-order behavior; they are not asserted to have finite rank.

Together with the finite closures described in
Sec.~\ref{sec:finite-classes-background}, this establishes both sides of the
one-mode polynomial classification. Gaussian dynamics factor through the
conditional means and covariances, while every fixed polynomial observable
under CMM dynamics factors through finitely many conditional-moment
coefficients \cite{emersonCMM}. Outside these two classes, the theorem gives infinite observable
rank for almost every Gaussian preparation and therefore rules out any global
finite-dimensional \(C^1\) realization of the stated observable response.

The proof has four steps.  First, a $d$-dimensional filter forces every finite
response matrix to have rank at most $d$.  Second, derivatives with respect to
the past record produce a sequence of state perturbations, while future time
derivatives produce a sequence of tests of those perturbations.  Third, every
Hamiltonian outside the Gaussian and CMM classes contains nonlinear coupling
that raises the leading polynomial degree of selected perturbations without
bound.  We choose rows and columns that detect these successive leading terms,
giving triangular response matrices with nonzero diagonal.  Finally, the few
places where the leading coefficient or a Gaussian pairing can vanish are
handled by changing the selected row, column, or constant reference input.
The result is a nonzero response minor of every size.

The rest of this section carries out these steps at the level needed to see
why the classification works.  The exact polynomial ordering, domain
statements, and exhaustive coefficient calculations are given in the
Supplemental Material \cite{supplemental}.

\subsection{From response derivatives to rank}
\label{sec:rank-tests}

The first point is to separate the size of a formal hierarchy from the rank
of the response.  An infinite list of moments can still depend on only a few
independent directions.  What must be shown is instead that perturbations of
the past record create arbitrarily many linearly independent future
responses.

If a $d$-dimensional filter produced the same input--output map, its
linearization would pass through the $d$-dimensional tangent state at the
splitting time.  The kernel would therefore have a separated form
\begin{equation}
 K_O(t,s)=\sum_{a=1}^d f_a(t)g_a(s).
 \label{eq:main-separated-kernel}
\end{equation}
Here the functions $g_a$ describe how a past perturbation changes the
internal state, and the functions $f_a$ describe how that change appears in
the future observable.

For the local calculation, take one-sided derivatives at the split,
\(s_0=t_0=\tau\), and form
\begin{equation}
 R_{ij}
 =
 \left.
 \partial_t^i\partial_s^jK_O(t,s)
 \right|_{(\tau,\tau)}.
 \label{eq:main-response-matrix}
\end{equation}
If Eq.~\eqref{eq:main-separated-kernel} held, then
\begin{equation}
 R_{ij}
 =\sum_{a=1}^d f_a^{(i)}(\tau)g_a^{(j)}(\tau),
\end{equation}
so every finite submatrix of \(R\) would have rank at most \(d\). To see
this directly, for each \(a\), the values \(f_a^{(i)}(\tau)\) form one
future-derivative vector and the values \(g_a^{(j)}(\tau)\) form one
past-derivative vector. Their contribution to \(R\) is the outer product of
these two vectors and therefore has rank one. Thus, \(d\) separated terms in
\(K_O\) produce a sum of at most \(d\) rank-one matrices. A nonzero \(N\times N\) minor
therefore proves \(\rank K_O\geq N\). Finding such minors for arbitrarily
large \(N\) proves infinite rank.

The rows and columns have a simple meaning in the controlled Wigner
evolution.  Take a constant reference input \(u_0\), write its Wigner-space evolution generator as
\(A=A_{u_0}\), and let \(W_\tau\) be the unnormalized state at the split.
For \(s<\tau<t\), the exact kernel can be separated at the splitting time
into a future readout factor and a past perturbation factor. A variation of
the record at time \(s\) acts on the state through the measurement operator
\({\cal M}_\eta\). Propagating the resulting perturbation forward to
\(\tau\), while pulling the centered observable backward from \(t\) to
\(\tau\), defines
\begin{align}
 F(t)
 &=
 e^{A^\dagger(t-\tau)}
 \bigl(o-\langle O\rangle_t\bigr),
 \label{eq:main-local-future-factor}\\
 G(s)
 &=
 e^{A(\tau-s)}{\cal M}_\eta e^{-A(\tau-s)}W_\tau .
 \label{eq:main-local-past-factor}
\end{align}
Both factors are therefore represented at the splitting time, where their
pairing gives
\begin{equation}
 K_O(t,s)
 =
 \frac{\langle F(t),G(s)\rangle}
 {\langle1,W_t\rangle},
 \qquad
 W_t=e^{A(t-\tau)}W_\tau .
 \label{eq:main-local-kernel-split}
\end{equation}
The factor \(F(t)\) is the future observable test transported back to the
split, while \(G(s)\) is the state direction obtained by inserting the past
record perturbation at time \(s\) and transporting it forward to the split.

Taking the \(i\)th derivative with respect to the future time \(t\) and the
\(j\)th derivative with respect to the past time \(s\) gives the row and column objects
\begin{align}
 F_i
 &=
 (A^\dagger)^i
 \bigl(o-\langle O\rangle_\tau\bigr),
 \label{eq:main-future-tests}\\
 G_j
 &=
 \operatorname{ad}_{-A}^{\,j}({\cal M}_\eta)W_\tau ,
 \label{eq:main-past-directions}
\end{align}
where
\[
 \operatorname{ad}_{-A}^{\,0}({\cal M}_\eta)={\cal M}_\eta,
 \qquad
 \operatorname{ad}_{-A}^{\,j+1}({\cal M}_\eta)
 =
 [-A,\operatorname{ad}_{-A}^{\,j}({\cal M}_\eta)] .
\]
The formula for \(G_j\) follows by differentiating the conjugated insertion
\(e^{A(\tau-s)}{\cal M}_\eta e^{-A(\tau-s)}\) with respect to \(s\) at
\(s=\tau\).  Thus \(G_j\) is the \(j\)th state direction created by moving the
past record insertion, while \(F_i\) is the \(i\)th future test applied to
that direction.

With these definitions, the leading part of the response-derivative matrix is
\begin{equation}
 R_{ij}
 =
 \frac{1}{\langle1,W_\tau\rangle}
 \langle F_i,G_j\rangle
 +\text{lower-order row--column terms}.
 \label{eq:main-response-pairing}
\end{equation}
At a high level, the proof now becomes a controlled choice of rows and
columns.  We choose past directions \(G_{j_c}\) and future tests
\(F_{i_r}\) so that, after ordering them by their leading polynomial
signature, the selected \(N\times N\) pairing matrix has a new highest-degree
term on each diagonal entry and only previously detected terms on one side of
the diagonal.  This makes the selected minor triangular at leading order, with
nonzero diagonal, and therefore nonzero determinant.

The lower-order terms above come from differentiating
\(\langle1,W_t\rangle^{-1}\) and the scalar centering
\(\langle O\rangle_t\) in Eq.~\eqref{eq:main-local-kernel-split}.  They do
not create the leading polynomial signatures used below.  In the selected
minors, elementary triangular elimination removes these lower-order
combinations without changing whether the determinant vanishes.  We may
therefore follow the highest polynomial signature in each pairing.  The
normalization and elimination steps are detailed in the Supplemental Material
\cite{supplemental}.

\subsection{How nonlinear mixing raises the response degree}
\label{sec:nonlinear-mixing-mechanism}

The leading Hamiltonian action makes the source of increasing rank visible.
The leading part of the Moyal commutator is the Hamiltonian phase-space flow,
\begin{equation}
 {\cal H}^{(1)}_hW
 =(\partial_xh)\partial_pW-(\partial_ph)\partial_xW,
\end{equation}
while the remaining Moyal terms contain more derivatives and therefore have
lower polynomial degree.  For a Hamiltonian monomial
$h_{\ell k}x^\ell p^k$, its action on a test
monomial begins with
\begin{equation}
 \begin{aligned}
 {\cal H}_{h_{\ell k}x^\ell p^k}^{\dagger}(x^ap^b)
 &=h_{\ell k}(ak-b\ell)\,
 x^{a+\ell-1}p^{b+k-1}\\
 &\quad+\text{terms of lower degree}.
 \end{aligned}
 \label{eq:main-leading-monomial-action}
\end{equation}
The displayed term is the ordinary Hamiltonian flow.  If
$\ell+k\ge3$ and $ak-b\ell\ne0$, one application raises the total degree
by $\ell+k-2$.  Repeated appearances of this term can therefore create an
unbounded sequence of distinct polynomial responses.

The coefficient \(ak-b\ell\) also explains why the proof must select
appropriate rows and columns rather than simply take consecutive ones. It
vanishes when the exponent pairs \((a,b)\) and \((\ell,k)\) are proportional. In
that case, we choose the first past direction \(G_j\), or future test \(F_i\),
whose leading exponent is not parallel to \((\ell,k)\). If the corresponding
pairing with the centered Gaussian vanishes by parity, we instead take an
additional past or future derivative with the opposite parity. Because
\(A_{u_0}\) contains \(u_0{\cal M}_\eta\), which adds one power of the measured
coordinate, choosing \(u_0\ne0\) makes both parities available. These choices
do not alter the degree-raising mechanism; they only select a row and column
in which it is visible.

The simplest complete cycle occurs when a nonlinear potential is combined
with a quadratic kinetic term,
\begin{equation}
 \begin{aligned}
 h(x,p)&=\frac{\beta}{2}p^2+V(x),\\
 V(x)&=v_dx^d+\text{lower powers},
 \qquad d\ge3.
 \end{aligned}
 \label{eq:main-kinetic-potential-symbol}
\end{equation}
At leading order the two pieces act as
\begin{align}
 x^ap^b
 &\xmapsto{\,p^2\,}
 \beta a\,x^{a-1}p^{b+1},
 \label{eq:main-kinetic-step}\\
 x^ap^b
 &\xmapsto{\,V(x)\,}
 -dv_db\,x^{a+d-1}p^{b-1}.
 \label{eq:main-potential-step}
\end{align}
Starting from a pure $x$ term, the kinetic part creates one power of $p$ and
the potential part converts it back into a higher power of $x$:
\begin{equation}
 x^{q_r}
 \longmapsto
 x^{q_r-1}p
 \longmapsto
 x^{q_r+d-2},
 \qquad
 q_r=1+r(d-2).
 \label{eq:main-degree-cycle}
\end{equation}
Neither piece alone is obstructed: the kinetic term is Gaussian and the
potential is in the aligned CMM class.  Their alternation transfers degree
between the measured and conjugate quadratures and raises it without bound.
This is the basic mixing mechanism behind the theorem.

For any $N$, we now select columns
$G_{j_1},\ldots,G_{j_N}$ along an increasing chain and rows
$F_{i_1},\ldots,F_{i_N}$ that detect its successive leading monomials.  On a
Gaussian preparation, every $G_j$ is a Gaussian weight times a polynomial.
The polynomial has a highest Hermite component fixed by its leading monomial,
and components of lower degree cannot contribute to that component.  Ordering
the selected rows and columns by this Hermite degree and subtracting earlier
ones from later ones leaves
\begin{equation}
 \langle F_{i_r},G_{j_c}\rangle=0\quad(r<c),
 \qquad
 \langle F_{i_r},G_{j_r}\rangle\ne0.
 \label{eq:main-triangular-pairings}
\end{equation}
The selected pairing matrix is therefore triangular with nonzero diagonal.
Its determinant is nonzero, so it gives an $N\times N$ response minor.  Since
$N$ is arbitrary, any Hamiltonian producing such a chain has infinite
response rank.

\subsection{Why the mechanism covers every Hamiltonian outside the finite classes}
\label{sec:mechanism-to-classification}

For \(X\)-homodyne monitoring, a polynomial depending only on \(X\) belongs
to the CMM class, while every polynomial of total degree at most two is
Gaussian.  The individual monomials belonging to neither class are therefore
\begin{equation}
 p^k\quad(k\geq3),
 \qquad
 x^\ell p^k
 \quad(\ell,k\geq1,\ \ell+k\geq3).
 \label{eq:main-offending-signatures}
\end{equation}

This list does not by itself cover every Hamiltonian outside
\({\mathscr F}_X\), because \({\mathscr F}_X\) is a union rather than a linear
space.  In particular,
\begin{equation}
 h(x,p)=\frac{\beta}{2}p^2+V(x),
 \qquad
 \beta\ne0,\quad \deg V\geq3,
 \label{eq:main-offending-binomial}
\end{equation}
contains no offending monomial: its kinetic term is Gaussian and its potential
term is CMM.  Their combination nevertheless belongs to neither finite class.
This is the kinetic--potential mechanism described in
Eq.~\eqref{eq:main-degree-cycle}.

For a general polynomial containing one of the monomials in
Eq.~\eqref{eq:main-offending-signatures}, a finite ordering of the powers of
\(x\) and \(p\) selects the highest offending term.  Its action in
Eq.~\eqref{eq:main-leading-monomial-action} supplies the leading entry in each
chosen column.  Lower Hamiltonian terms, damping, diffusion, and the
It\^o--Stratonovich correction have lower order in the same arrangement.
They may change entries below the selected diagonal, but they cannot cancel
its leading coefficient.

The only remaining family for which this ordering does not immediately finish
the argument is
\begin{equation}
 h(x,p)=\frac{\beta}{2}p^2+F(x)p+V(x),
 \qquad \beta\ne0.
 \label{eq:main-residual-shear-family}
\end{equation}
Completing the square gives
\begin{equation}
 h(x,p)=\frac{\beta}{2}
 \left(p+\frac{F(x)}{\beta}\right)^2
 +V(x)-\frac{F(x)^2}{2\beta}.
 \label{eq:main-completed-square}
\end{equation}
The binomial in Eq.~\eqref{eq:main-offending-binomial} is the case \(F=0\).
More generally, the shift of \(p\) leaves the measured coordinate \(x\)
unchanged.  If the remaining potential is nonlinear, the
kinetic--potential cycle \eqref{eq:main-degree-cycle} applies in the shifted
coordinate.  If its nonlinear part cancels, then \(F\) must be nonlinear,
since otherwise the original Hamiltonian would be Gaussian.  The momentum
shift places this nonlinearity in the record and damping operators.  Selecting
columns from the resulting \(G_j\) sequence, rather than from the cancelled
potential sequence, again produces unbounded leading degrees.  This is the
only cancellation that requires a different sequence of columns.

Together with the case treated above, the offending monomials and this
residual family exhaust all polynomial Hamiltonians outside the Gaussian and
CMM classes.  For every \(N\), the corresponding
\(N\times N\) response determinant is nonzero for a coherent Gaussian
preparation and a suitable constant input \(u_0\).  Each determinant is an
analytic function of the Gaussian parameters, so its zero set has measure
zero.  Taking the countable union of these zero sets over
\(N=1,2,\ldots\) still gives a measure-zero set.  Outside it, nonzero response
minors exist at every size, proving
Theorem~\ref{thm:main-classification}.

\subsection{Robustness and extensions}
\label{sec:classification-extensions}

For any proposed dimension \(d\), the selected
\((d+1)\times(d+1)\) minor varies continuously with the reference input.
Once nonzero, it remains nonzero on an open record neighborhood.  Because
every such neighborhood has positive probability under the physical record
law, the obstruction rules out both a robust local realization and a
\(C^1\) filter claimed to be exact almost surely, rather than only a filter at
one finely tuned record.

Additionally, the result rotates with the monitored quadrature.  For
\begin{equation}
 X_\theta=X\cos\theta+P\sin\theta,
 \qquad
 P_\theta=-X\sin\theta+P\cos\theta,
\end{equation}
the finite class is
\begin{equation}
 {\mathscr F}_\theta={\mathscr G}\cup
 \left\{V(X_\theta)+aP_\theta
 +b\operatorname{Sym}(X_\theta P_\theta)\right\}.
 \label{eq:main-rotated-finite-class}
\end{equation}
Every polynomial Hamiltonian outside \({\mathscr F}_\theta\) has infinite
rank under \(X_\theta\)-homodyne monitoring.

For heterodyne monitoring, the past perturbation has two real components, so
\begin{equation}
 K_O:L^2(I_-;\mathbb R^2)\longrightarrow L^2(I_+).
 \label{eq:main-heterodyne-KO}
\end{equation}
If this operator had finite rank, its restriction to either record channel
would also have finite rank.  A nonlinear Hamiltonian can be aligned with at
most one of the two conjugate measurement directions, so the other channel
produces the one-mode obstruction.  The finite heterodyne polynomial class
therefore reduces to the Gaussian class.

Inefficient detection with \(\eta>0\) rescales the record insertion, while
thermal diffusion and damping leave the protected leading terms unchanged.
These effects modify the response coefficients but not the exact rank.

The single-mode obstruction also persists when the mode is embedded in a
finite network with generic polynomial coupling.  At zero coupling the
selected minors reduce to the nonzero one-mode minors, and their analytic
dependence on the coupling parameters keeps them nonzero away from a
measure-zero exceptional set.  This establishes that the one-mode boundary is stable
under generic finite-network embeddings without requiring a separate
classification of every multimode nonlinearity. As noted at the beginning of this section, precise statements and proofs of
these results are provided in the Supplemental Material
\cite{supplemental}.


\section{Approximation from the response spectrum}
\label{sec:approximation}

The exact classification distinguishes finite from infinite response rank, but
does not determine how well an infinite-rank response can be approximated.
Its independent response directions may decay rapidly, leaving only a small
number relevant at a chosen accuracy. We study this question at three levels.
First, the singular values of \(K_O\) define the optimal local response
dimension. We then test whether the corresponding coordinates remain useful
for finite changes of the record. Finally, we compare this local benchmark
with a global approximation constructed specifically for Kerr dynamics. This section develops the mathematical constructions, while the following
section presents their numerical implementation and results.

\subsection{Optimal local response approximation}

Fix a physical record and the past and future intervals \(I_-\) and \(I_+\)
used to define \(K_O\).  For a small deterministic perturbation
\(h\in L^2(I_-)\) of the past record,
\begin{equation}
 \delta\langle O\rangle(t)
 =
 (K_Oh)(t)+o(\|h\|_{L^2(I_-)}),
 \qquad t\in I_+.
 \label{eq:main-linear-response}
\end{equation}
On the finite intervals considered here, the response kernel
\(K_O(t,s)\) is square integrable, so \(K_O\) is a Hilbert--Schmidt operator
and therefore compact.  It consequently admits a singular-value
decomposition,
\begin{equation}
 K_Oh
 =
 \sum_{j\geq1}
 \sigma_j\langle v_j,h\rangle_{I_-}w_j,
 \label{eq:main-response-svd}
\end{equation}
where \(\{v_j\}\) and \(\{w_j\}\) are orthonormal singular functions on
\(I_-\) and \(I_+\), respectively, and
\(\sigma_1\geq\sigma_2\geq\cdots\geq0\). The inner product
$\langle v_j,h\rangle_{I_-}$ is simply the amount of the $j$th past pattern
contained in the applied perturbation.  A component $v_j$ is therefore sent
to the future component $w_j$ with strength $\sigma_j$.

Keeping the first $d$ terms gives
\begin{equation}
 K_O^{(d)}h
 =\sum_{j=1}^{d}
 \sigma_j\langle v_j,h\rangle_{I_-}w_j.
 \label{eq:main-truncated-response}
\end{equation}
This is the best rank-\(d\) linear approximation to \(K_O\)
\cite{eckartyoung1936}. Related uses of Hankel singular values and balanced
coordinates in nonlinear model reduction are developed in
Refs.~\cite{scherpen1993,fujimotoscherpen2010}.  In operator norm,
\begin{equation}
 \inf_{\operatorname{rank}R\leq d}\|K_O-R\|
 =\sigma_{d+1},
 \label{eq:main-rank-d-error}
\end{equation}
while in Hilbert--Schmidt norm the optimal error is
$\bigl(\sum_{j>d}\sigma_j^2\bigr)^{1/2}$.  We use the operator-norm
definition
\begin{equation}
 d_\epsilon
 =\min\left\{d:
 \frac{\sigma_{d+1}}{\sigma_1}\leq\epsilon
 \right\}.
 \label{eq:main-effective-dimension}
\end{equation}
Thus $d_\epsilon$ is the smallest number of record directions that reproduces
every infinitesimal perturbation to relative response error at most
$\epsilon$.  A finite-rank system has $\sigma_j=0$ beyond some index.  An
infinite-rank system has no exact termination, but it can still have a small
$d_\epsilon$ if its singular values decrease rapidly.  This distinction
between exact and finite-accuracy memory is what Fig.~\ref{fig:main-response-complexity}
tests first.

\subsection{Finite changes of a record}

The SVD describes the first-order response around one reference record.  To
test its accuracy for small but noninfinitesimal changes, we perturb the past
record along the first \(d\) singular directions,
\begin{equation}
 h(t)=\sum_{j=1}^{d}z_jv_j(t),
 \qquad
 z_j=\langle v_j,h\rangle_{I_-}.
 \label{eq:main-hankel-coordinates}
\end{equation}
The coefficients \(z_j\) specify how much of each past response pattern is
added.  Because the \(v_j\) are orthonormal, the total perturbation radius is
\begin{equation}
 r=\|h\|_{L^2(I_-)}
 =\left(\sum_{j=1}^{d}z_j^2\right)^{1/2}.
 \label{eq:main-perturbation-radius}
\end{equation}
In the simulations, this perturbation is applied to the observation record as
\begin{equation}
 dY_t^{(h)}=dY_t+h(t)\,dt.
\end{equation}
Thus \(r\) measures the integrated magnitude of the deterministic drift added
over the past window.

The linear prediction in these coordinates is fixed by the response operator:
\begin{equation}
 \delta\langle O\rangle_{1}(t;z)
 =\sum_{j=1}^{d}\sigma_jz_jw_j(t).
 \label{eq:main-linear-hankel-decoder}
\end{equation}
For a noninfinitesimal perturbation, the full past-to-future map is generally
curved.  We approximate its leading curvature on the same \(d\)-dimensional
record subspace by writing
\begin{align}
 \delta\langle O\rangle_{2}(t;z)
 &=\delta\langle O\rangle_{1}(t;z)
   +\sum_{1\leq j\leq k\leq d}A_{jk}(t)z_jz_k,
 \label{eq:main-quadratic-hankel-decoder}\\
 \delta\langle O\rangle_{3}(t;z)
 &=\delta\langle O\rangle_{2}(t;z)
   +\sum_{1\leq j\leq k\leq\ell\leq d}
    B_{jk\ell}(t)z_jz_kz_\ell.
 \label{eq:main-cubic-hankel-decoder}
\end{align}
The functions \(A_{jk}(t)\) and \(B_{jk\ell}(t)\) are fitted using full
simulations at symmetric perturbations around the reference record and tested
on separate perturbation directions.

The SVD is optimal only for the linearization \(K_O\).  The quadratic and
cubic models ask whether the coordinates selected by that linearization remain
useful for the nearby nonlinear map.  They contain more fitted coefficients,
but they still receive only the same \(d\) inputs \(z_1,\ldots,z_d\) and
approximate the past-to-future map only on
\[
 u+\operatorname{span}\{v_1,\ldots,v_d\}.
\]
They therefore cannot recover a record direction omitted from this subspace.
The
construction here answers the narrower question of whether the dimension
selected from \(K_O\) remains useful for nearby finite perturbations.
Figure~\ref{fig:main-local-attainability} tests this question.

For a perturbation of radius \(r\), we compare the exact and predicted future
changes using the scaled trajectory error
\begin{equation}
 E(r)=
 \frac{
 \|\delta\langle O\rangle_{\mathrm{exact}}
 -\delta\langle O\rangle_{\mathrm{model}}\|_{L^2(I_+)}
 }{r\sigma_1}.
 \label{eq:main-scaled-local-error}
\end{equation}
The numerator is the \(L^2\) error over the full future trajectory, rather than
an error at one time.  The denominator is the largest linear
\(L^2(I_+)\) response possible from a perturbation of radius \(r\).  This is
the same response scale used to define \(d_{10^{-3}}\), so
\(E(r)=10^{-3}\) can be compared directly with that threshold.  

\subsection{A structured Kerr approximation}
\label{sec:constructive-state-approximations}

The singular values of \(K_O\) describe the smallest local response dimension,
but they do not by themselves provide evolution equations for a causal filter.
For a constructive comparison, we introduce a finite-dimensional approximation
adapted specifically to the monitored Kerr oscillator.  Unlike the local
response models above, this approximation evolves directly with the
measurement record and can be compared with full Fock integration.

The Kerr Hamiltonian is diagonal in photon number, and a coherent state has a
Poisson photon-number distribution.  This suggests expanding the conditional
state in orthogonal polynomials adapted to a Poisson weight rather than
retaining a separate amplitude for every occupied Fock level.  Damping changes
the mean photon number from \(N_0\) to \(N_0e^{-\kappa t}\).  For a fixed
filtering horizon \(T\), we use the reference weight
\begin{equation}
 w_{\mu}(n)=\frac{e^{-\mu}\mu^n}{n!},
 \qquad
 \mu=N_0e^{-\kappa T}.
 \label{eq:main-charlier-reference-weight}
\end{equation}
Let \(\phi_r(n;\mu)\) denote the Charlier polynomials orthonormal in this
weight \cite{koekoekleskyswarttouw2010},
\begin{equation}
 \sum_{n\geq0}
 w_\mu(n)\phi_r(n;\mu)\phi_s(n;\mu)
 =\delta_{rs}.
 \label{eq:main-charlier-orthogonality}
\end{equation}
Projection onto the first \(M+1\) polynomials gives the best degree-\(M\)
approximation in the corresponding weighted norm.  The Poisson weight carries
the broad photon-number envelope, while the polynomial coefficients describe
the remaining number dependence.

To obtain the approximation, remove the known coherent damping envelope and
diagonal Kerr phase from the number amplitudes and denote the remaining number
dependence by \(U_n(t)\).  Under ideal heterodyne monitoring, it obeys
\begin{equation}
 \begin{aligned}
 dU_n(t)
 &=\gamma(t)e^{-iKtn}U_{n+1}(t)\,dZ_t^*,\\
 \gamma(t)
 &=\sqrt{\kappa}\,\alpha_0e^{-\kappa t/2},
 \qquad N_0=|\alpha_0|^2,
 \end{aligned}
 \label{eq:main-kerr-residual-lattice}
\end{equation}
where \(Z_t=(Y_{1,t}+iY_{2,t})/\sqrt{2}\) and
\(dZ_t\,dZ_t^*=dt\).  Expanding the residual function as
\begin{equation}
 U_n(t)=\sum_{s\geq0}c_s(t)\phi_s(n;\mu)
 \label{eq:main-charlier-expansion}
\end{equation}
and projecting onto the same orthonormal basis gives
\begin{align}
 dc_r(t)
 &=
 \gamma(t)\sum_{s\geq0}
 B_{rs}(t)c_s(t)\,dZ_t^*,
 \label{eq:main-charlier-coefficient-system}\\
 B_{rs}(t)
 &=
 \sum_{n\geq0}
 w_\mu(n)\phi_r(n;\mu)
 e^{-iKtn}\phi_s(n+1;\mu).
 \label{eq:main-charlier-matrix}
\end{align}
The matrices \(B(t)\) depend only on the model parameters and can be computed
before processing a record. For a coherent initial state,
\(c_0(0)=1\) and \(c_r(0)=0\) for \(r>0\).

Retaining \(r,s=0,\ldots,M\) gives a causal stochastic approximation with
\(M+1\) complex coefficients. Increasing \(M\) systematically enlarges the
represented polynomial space, and the numerical tests below compare its
accuracy with a well-converged Fock calculation.

The Charlier construction therefore provides a model-adapted alternative to
standard Fock truncation for monitored Kerr oscillators. It also provides a constructive comparison with the effective response
dimensions obtained from \(K_O\), but, as we will see, remains far from the
local Hankel optimum.


\section{Numerical realization tests}
\label{sec:numerics}

With these constructions in place, we now test how rapidly the response
spectrum decays beyond the exact finite boundary, whether its leading
directions remain predictive for finite changes around a physical record, and
how closely a causal global representation approaches the local optimum. As infinite-rank
examples, we consider the experimentally relevant Kerr and Duffing Hamiltonians
\begin{equation}
 \begin{aligned}
 H_{\mathrm K}
 &=\frac{K}{2}N(N-1),
 \qquad N=a^\dagger a,\\
 H_{\mathrm D}
 &=\frac12(P^2+X^2)+\frac{g}{4}X^4.
 \end{aligned}
 \label{eq:main-numerical-hamiltonians}
\end{equation}
The former is nonlinear in both quadratures. In the latter, the kinetic term
repeatedly mixes the \(X\)-dependent nonlinearity with \(P\). Both therefore
lie outside the finite aligned-\(X\)-homodyne classes when their nonlinear
coefficients are nonzero.

All simulations use \(\kappa=1\), with time measured in units of
\(1/\kappa\). The sampled kernels are weighted using the \(L^2\) quadrature
before taking their SVD, so that the resulting spectra approximate those of
the continuous operator \(K_O:L^2(I_-)\to L^2(I_+)\). Numerical algorithms,
convergence tests, and complete parameter sets are given in the Supplemental
Material, while fixed seeds, record-level results, and figure scripts are
provided in the reproducibility archive
\cite{supplemental,emerson_reproducibility_2026}. The physical parameters
needed to interpret each experiment are stated below.

\subsection{Exact rank and finite-accuracy memory}

\begin{figure*}[t]
\centering
\includegraphics[width=\textwidth]{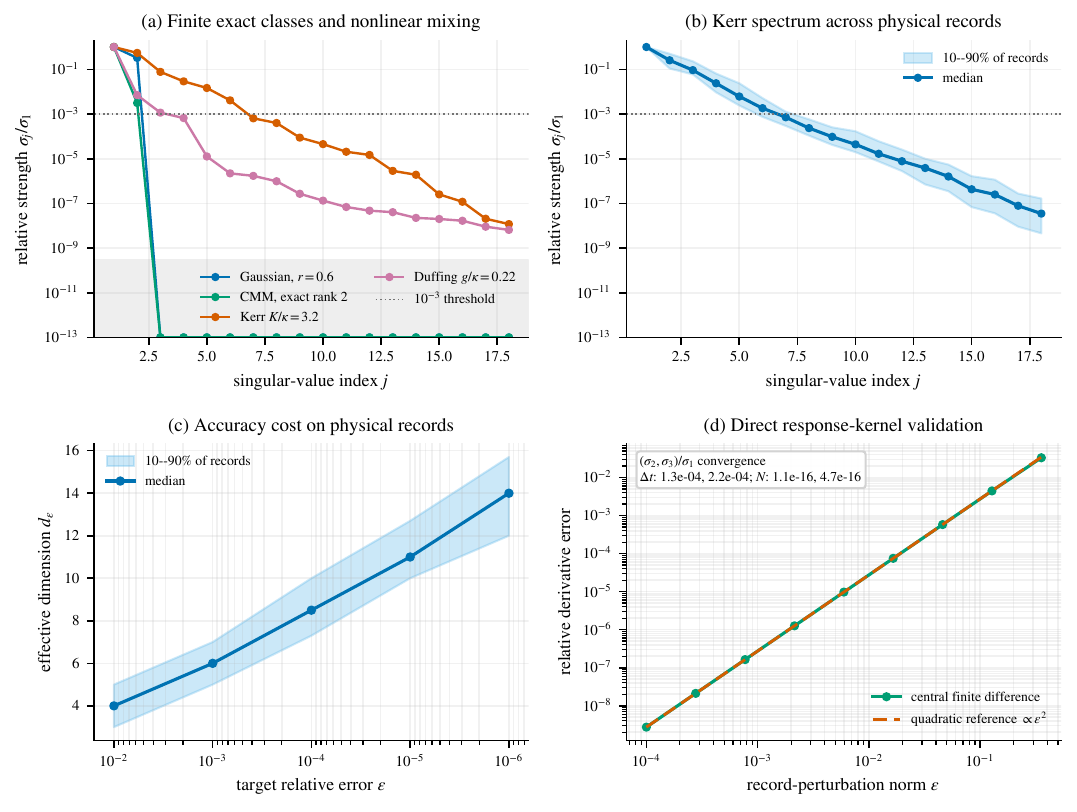}
\caption{Exact and effective observable memory on physical records.  (a)
Relative singular spectra for two finite exact controls and two nonlinear
mixing systems.  The Gaussian and CMM responses terminate after two
directions, whereas Kerr and Duffing retain decaying tails.  (b) Median and
10th--90th percentile Kerr spectra over 24 physical homodyne records.  (c)
The corresponding effective dimension $d_\epsilon$.  (d) Direct
central-finite-difference validation of the discrete response kernel.  All
panels use ideal monitoring and separated past and future windows.  Exact
zeros of the analytic CMM kernel are shown at the plotting floor.}
\label{fig:main-response-complexity}
\end{figure*}

Figure~\ref{fig:main-response-complexity}(a) first checks that the spectrum
distinguishes the two sides of the theorem.  The Gaussian control uses
$H=0.85N$, $O=X$, and a displaced squeezed state with squeezing $r=0.6$.
Its two response directions are the two components of the conditional mean;
the covariance obeys a record-independent equation.  A coherent initial
state would give a degenerate vanishing $X$ response, so the squeezed input is
used to display the nonzero Gaussian rank clearly.

All four curves in panel (a) use past and future windows of length \(0.55\),
separated by a gap of \(0.06\).  The displacement is
\(\alpha=1.4+0.55i\).  The nonlinear examples
begin in \(\lvert\alpha\rangle\), while the Gaussian control begins in
\(D(\alpha)S(0.6)\lvert0\rangle\).  The same Kerr parameters are used for the
24-record ensemble in panels (b) and (c).

The nonlinear finite control uses aligned $X$ homodyne monitoring with
\begin{equation}
 H=0.10X^3,
 \qquad O=P.
\end{equation}
Its continuum response kernel has exact rank two, as shown by
Eq.~\eqref{eq:main-cmm-explicit-kernel}.  The two leading numerical directions
agree with that kernel; the analytic zeros from the third singular value
onward are shown at the common plotting floor.  The apparent numerical tail
converges toward zero under time-step refinement, as detailed in the
Supplemental Material \cite{supplemental}.

The obstructed curves use a strong $K/\kappa=3.2$ for Kerr and $g/\kappa=0.22$ for
Duffing, with $O=X$.  They do not terminate at a small index, and their
leading tails persist under numerical refinement.  This is the finite
calculation expected from the theorem's infinite-rank conclusion.  Just as
importantly, both tails decay.  Infinite exact rank therefore does not mean
that accurate approximation requires a large number of directions at every
tolerance.

Panel (b) repeats the Kerr calculation over 24 independently generated
physical homodyne records.  The band shows record-to-record variation, not a
confidence interval or a numerical-error estimate.  The decaying tail is
present throughout the ensemble rather than being selected from one favorable
record.  Panel (c) converts each spectrum into $d_\epsilon$ using
Eq.~\eqref{eq:main-effective-dimension}.  At $\epsilon=10^{-3}$ the median is
six and the observed values range from five to seven.  Thus the response is
infinite rank exactly, while about six past-record directions capture its
local effect on $X$ to relative operator-norm accuracy $10^{-3}$ in this
experiment.

Panel (d) validates the tangent construction of the sampled kernel against
central differences of the complete simulated record-to-observable map.  The
relative discrepancy follows the expected \(\varepsilon^2\) scaling and is
\(2.74\times10^{-9}\) at \(\varepsilon=10^{-4}\).  The tangent propagation
and validation procedure are given in the Supplemental Material
\cite{supplemental}.

\subsection{Finite-amplitude prediction in the response coordinates}

Figure~\ref{fig:main-local-attainability} tests Eqs.~\eqref{eq:main-linear-hankel-decoder}--
\eqref{eq:main-cubic-hankel-decoder}.  For each physical record, we choose
\(d=d_{10^{-3}}\) and fit the quadratic and cubic terms using symmetric
perturbations of the retained coordinates.  The linear term is fixed by the
SVD, and all plotted errors use independently drawn test directions that can
also contain omitted response components.  The fitting and testing procedure
is detailed in the Supplemental Material \cite{supplemental}.

\begin{figure*}[t]
\centering
\includegraphics[width=\textwidth]{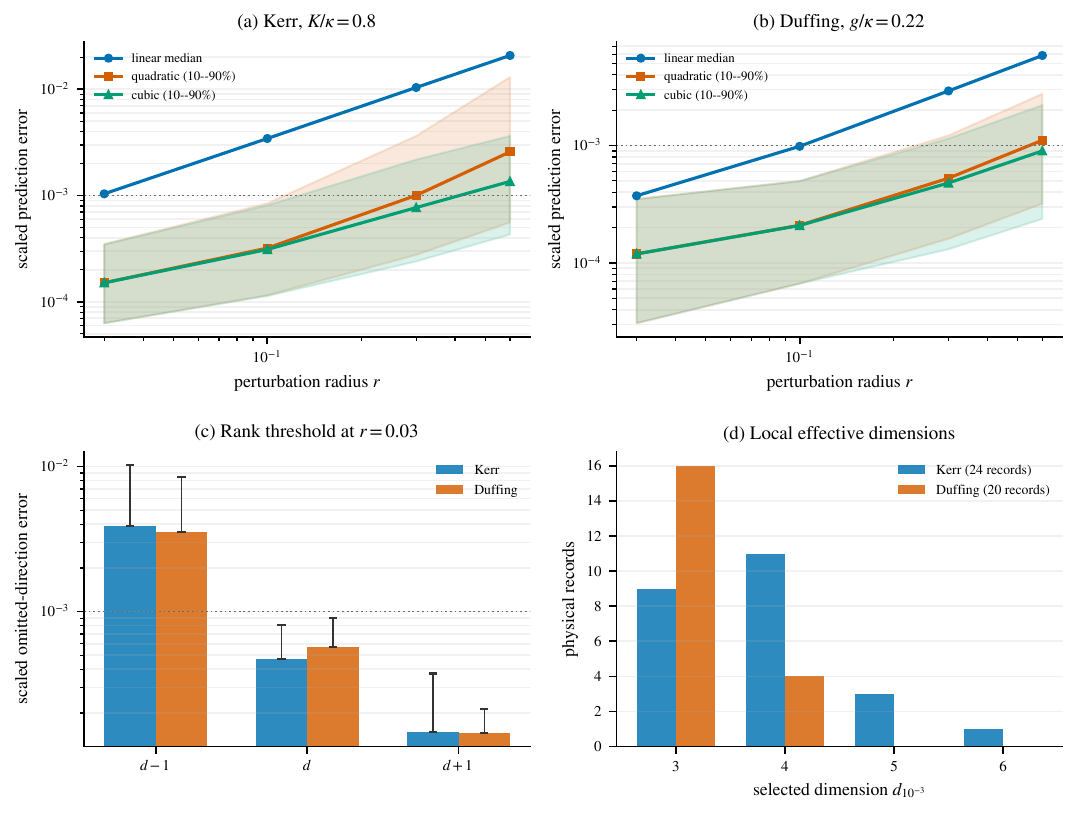}
\caption{Finite prediction in the Hankel-selected coordinates.  Median scaled
trajectory error \(E(r)\) versus \(L^2\) record-perturbation radius \(r\) for (a) Kerr
and (b) Duffing physical records.  The quadratic and cubic curves use the
same \(d=d_{10^{-3}}\) coordinates as the linear response; their shading spans
the 10th--90th percentiles, while the linear curve is a median-only baseline.
(c) Response along the first omitted direction when \(d-1\), \(d\), or
\(d+1\) directions are retained at \(r=0.03\); upper error bars reach the
90th percentile.  (d) Empirical counts of the independently selected local
dimensions over 24 Kerr and 20 Duffing records.  The dotted line marks the
target error \(10^{-3}\).}
\label{fig:main-local-attainability}
\end{figure*}

Panel (a) uses \(K/\kappa=0.8\) and \(O=X\) over 24 Kerr records.  The linear
prediction is accurate only very near the reference trajectory because its
leading missing term is quadratic in the perturbation.  Adding quadratic
dependence on the same coordinates removes most of this curvature.  Cubic
terms make little difference at the smallest radii but improve the prediction
farther from the reference record.  At \(r=0.1\), the cubic median error is
\(3.11\times10^{-4}\), and it remains below \(10^{-3}\) through \(r=0.3\). 
The initial coherent amplitude is \(2\).

Panel (b) repeats the construction for Duffing at \(g/\kappa=0.22\) over 20
records.  Quadratic and cubic decoding are again nearly identical close to
the reference record.  At \(r=0.6\), their median errors are respectively
\(1.11\times10^{-3}\) and \(9.00\times10^{-4}\), so cubic dependence becomes
useful only toward the edge of the tested neighborhood.  Agreement across
Kerr and Duffing shows that the construction is not tied to the number-basis
structure of the Kerr oscillator.  Here the initial coherent amplitude is
\(1.4+0.55i\).

Panel (c) tests whether the singular-value rule identifies a transition at
the requested accuracy rather than merely giving a sufficient dimension.  At
\(r=0.03\), each candidate dimension is tested along its first omitted
singular function.
The retained coordinates then vanish, so no quadratic or cubic combination of
them can reproduce the missing response.  With \(d-1\) directions, the median
errors are \(3.89\times10^{-3}\) for Kerr and \(3.51\times10^{-3}\) for
Duffing.  At \(d\), they fall to \(4.72\times10^{-4}\) and
\(5.74\times10^{-4}\), below the chosen target even at the 90th percentile;
\(d+1\) reduces them further.  The dimension selected from the derivative
therefore remains the relevant cutoff at small finite radius.

Panel (d) shows the record-to-record variation in the dimension selected by
the \(10^{-3}\) singular-value cutoff.  For Kerr, 20 of the 24 records require
three or four coordinates, three require five, and one requires six.  For
Duffing, 16 of the 20 records require three coordinates and the remaining four
require four.  These counts are an empirical robustness check rather than a
precise estimate of the probability of each dimension.  Together, the panels
show that the selected record space is small and fairly stable across records,
that one fewer direction misses the target accuracy, and that nonlinear
functions of the same coordinates remain useful beyond the strictly linear
regime.

The smaller Kerr dimensions here than in
Fig.~\ref{fig:main-response-complexity}(c) are not contradictory.  They
mainly reflect the weaker nonlinearity, \(K/\kappa=0.8\) rather than
\(K/\kappa=3.2\), together with the different window choices.

\subsection{A Poisson--Charlier state approximation for Kerr dynamics}

The final experiment changes from a local observable description to a
globally propagated state filter.  We integrate the truncated coefficient
system in Eq.~\eqref{eq:main-charlier-coefficient-system} under ideal heterodyne
monitoring and reconstruct the resulting state in a common Fock basis.  On
each physical record, the endpoint state error is
\begin{equation}
 e_{\mathrm{state}}
 =\min_{\varphi}
 \|\psi_{\mathrm{Fock}}-e^{i\varphi}\psi_M\|_2,
 \label{eq:main-charlier-state-error}
\end{equation}
where the phase minimization removes the physically irrelevant global phase.
For the response comparison, we instead use the operator norm of the
difference between the exact and Charlier response matrices, divided by
$\sigma_1$.

\begin{figure*}[t]
\centering
\includegraphics[width=\textwidth]{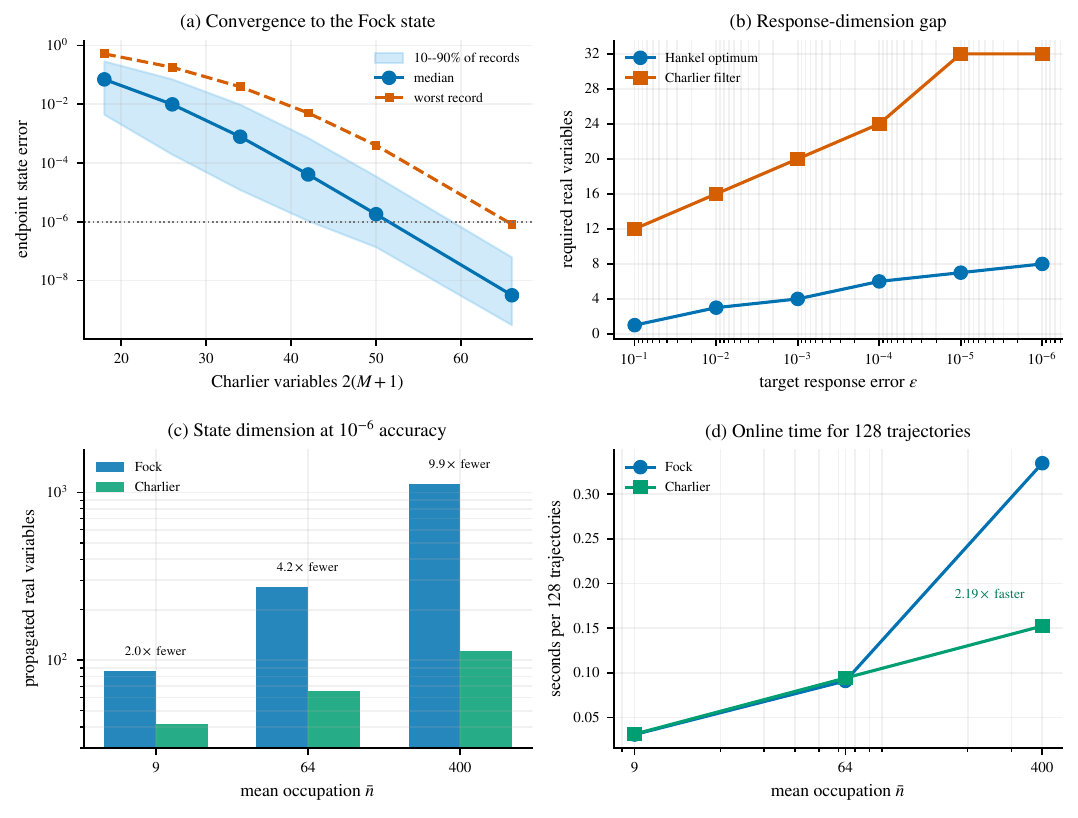}
\caption{Poisson--Charlier approximation of monitored Kerr dynamics.
(a) Phase-aligned endpoint state error over 40 physical heterodyne Kerr
records.  (b) Real variables required by the Charlier response and by the
optimal rank-$d$ response approximation on the same record.  (c)
Accuracy-matched Fock and Charlier state dimensions at fixed accumulated
mean-field phase.  (d) Their separately measured online times for a batch of
128 trajectories, with reusable operator construction excluded.  Panels
(c) and (d) use a common 90th-percentile state tolerance of $10^{-6}$.}
\label{fig:main-charlier-payoff}
\end{figure*}

Panel (a) uses $N_0=9$, $K/\kappa=0.4$, and $\kappa T=0.8$.  For each of 40
physical heterodyne records, the Charlier and cutoff-converged Fock equations
are driven by the same record.  Increasing the order from $M=8$ to $M=32$
reduces the 90th-percentile endpoint state error from
$2.86\times10^{-1}$ to $6.11\times10^{-8}$.  Every order-32 record lies below
$10^{-6}$, with worst error $8.12\times10^{-7}$.  Thus the error decreases
consistently across the tested orders and records rather than only after
averaging an observable over the ensemble.  The convergence and matched-record
integration checks are given in the Supplemental Material
\cite{supplemental}.

Panel (b) returns to the observable response.  For each requested error, it
compares the first tested Charlier dimension that reaches the target with the
optimal rank required to approximate the exact linearized response.  At
relative response error $10^{-3}$, the Hankel optimum
needs four real coordinates, whereas the Charlier sequence first reaches the
target at 20.  The difference is expected: the local SVD devotes every
coordinate to one future observable near one record, while the Charlier
filter preserves enough information to reconstruct the conditional state
throughout the measurement.  The Hankel curve is therefore a lower frontier,
not the runtime of an already constructed four-variable causal filter.
This same-record comparison uses $N_0=9$, $K/\kappa=0.4$, and $T=0.75$.

Panels (c) and (d) ask whether the global Charlier structure nevertheless
improves on direct Fock propagation.  The scan uses \(T=0.5\) and, to compare
different occupations without also increasing the accumulated mean-field
nonlinearity, fixes
\begin{equation}
 \Phi=
 \frac{KN_0(1-e^{-\kappa T})}{\kappa}=1,
 \label{eq:main-fixed-kerr-phase}
\end{equation}
and uses the smallest tested order whose 90th-percentile endpoint state error
is at most $10^{-6}$.  At occupations $N_0=9$, $64$, and $400$, the Charlier
representation uses respectively $42$, $66$, and $114$ real variables,
compared with $86$, $274$, and $1130$ for the matched Fock calculations.
The reduction therefore grows from a factor of $2.0$ to $4.2$ and then
$9.9$ across the tested cases.  These three points demonstrate increasing
compression over the tested range; they are not used to fit an asymptotic
scaling law.

Panel (d) plots the actual online propagation times rather than their ratio.
The Fock baseline is an optimized direct photon-number amplitude evolution,
not a generic dense master-equation solver.  Reusable operators are
precomputed for both methods and excluded from the timing.  At $N_0=400$, a
batch of 128 trajectories takes $0.335$ seconds in the direct Fock
representation and $0.152$ seconds in the Charlier representation, giving the
reported factor $2.19\times$.  At the two smaller occupations the times are close,
and single-trajectory timings remain comparable.  The calculation therefore
shows both the benefit and the limitation of state compression: fewer
variables reduce storage immediately, while a wall-clock advantage appears
only when the problem is large enough for that reduction to outweigh the
basis-update overhead.

\section{Discussion}
\label{sec:discussion}

The central result of this work is an exact realizability boundary for monitored nonlinear bosonic dynamics.  For one-mode polynomial Hamiltonians under quadrature monitoring, Gaussian dynamics and the measurement-aligned CMM class admit finite observable filters, whereas dynamics outside these classes generically have infinite observable Hankel rank.  The obstruction is coordinate independent and persists under inefficient detection, thermal noise, rotated monitoring, and generic finite-network embeddings.  The significance of this boundary is not that nonlinear filtering is generically impossible, but that outside the finite classes approximation becomes unavoidable: no alternative smooth parametrization can turn the exact record-to-observable map into a finite-dimensional filter.

The numerical results show that this exact distinction need not coincide with practical complexity.  The Gaussian and CMM controls terminate at finite rank as predicted, while Kerr and Duffing have nonterminating response spectra.  Yet the singular values of the latter decay rapidly over the regimes studied.  For the strongly nonlinear Kerr example, an exactly infinite-rank response requires only about six directions at relative local accuracy \(10^{-3}\), with similar behavior across independently generated physical records.  Thus infinite exact memory can coexist with small finite-accuracy memory.  We do not infer a universal decay law from two models, but the examples establish that the exact no-go result need not imply an impractically large description at experimentally relevant accuracy.

The finite-amplitude tests strengthen this interpretation.  The singular vectors are selected entirely from the derivative of the record-to-observable map at one physical trajectory, but nonlinear functions of the same coordinates remain predictive for appreciable finite changes of the record.  Moreover, the \(d-1\), \(d\), and \(d+1\) comparison shows that the singular-value cutoff is not merely a sufficient truncation: at small finite radius, omitting the next selected direction moves the error above the target, while retaining it restores the requested accuracy.  The linearized Hankel operator therefore identifies directions that remain relevant to the nearby nonlinear map.  This does not by itself construct a global filter, but it suggests that local response analysis can reveal the information-bearing coordinates from which one might be built.

The global Kerr calculation shows both the promise and the limitation of this viewpoint.  The Poisson--Charlier representation gives a genuine causal state approximation whose reduction relative to Fock propagation grows substantially with occupation and eventually produces a wall-clock advantage.  At the same time, it remains well above the local observable frontier: at relative response error \(10^{-3}\), the tested Charlier representation requires substantially more variables than the optimal local response rank.  These dimensions measure different tasks.  The Charlier filter must propagate enough information to reconstruct the conditional state over an evolving trajectory, whereas the Hankel rank concerns one future observable near one reference record.  The resulting gap may therefore contain two effects: avoidable overhead from retaining information irrelevant to the chosen output, and an intrinsic cost of transporting locally optimal information as the trajectory moves through record space.  Distinguishing these possibilities is the main constructive problem left by the present work.

This distinction also separates several notions of complexity that are usually conflated.  State dimension measures the resources used by a particular representation of the conditional state.  Exact observable memory asks for the minimum dimension through which the exact record-to-observable map can factor.  The singular spectrum adds a third quantity, finite-accuracy observable memory, which asks how many independent directions of the past record remain relevant to a specified future prediction at tolerance \(\epsilon\).  An infinite hierarchy of moments need not imply large exact memory, as the CMM examples demonstrate, and infinite exact memory need not imply large finite-accuracy memory, as the Kerr and Duffing spectra demonstrate.  Conversely, a representation efficient for the full state may retain much more information than is needed for a particular readout.

This suggests a broader way to formulate the computational complexity of monitored quantum systems. Rather than asking only for the resources required to simulate the full quantum state, one can ask for the resources required to predict a specified quantum readout from its preceding measurement record to a fixed accuracy \cite{huang2020predicting}. The effective dimension \(d_\epsilon\) provides a local, representation-independent measure of one part of this problem: the amount of past information that remains distinguishable through the chosen future observable. It is not by itself a computational-complexity measure, since a globally efficient algorithm must also transport these coordinates causally and with controlled update cost and numerical precision. Nevertheless, it provides a natural lower frontier against which such algorithms can be compared. Because this frontier is defined for a chosen output, different observables of the same monitored system may have different predictive complexity. Establishing when this local memory can be promoted to an efficient global predictor, and when additional global resources are unavoidable, could connect realization theory to complexity questions formulated directly at the level of quantum readouts.

Finally, the realization framework itself is not tied to bosonic phase space. The polynomial bosonic problem provides a setting in which the exact boundary can be evaluated analytically, but the past-to-future construction applies more generally whenever a differentiable quantum input--output map can be defined. Spin, fermionic, and hybrid systems therefore provide natural settings in which to ask whether complicated monitored many-body dynamics can nevertheless have simple observable readouts \cite{ladewig2022monitored,poggi2024measurement}. Coupled systems are especially interesting because the conditional state can become very high dimensional even when only one or a few observables are experimentally read out, creating the possibility that observable-level predictive complexity is much smaller than full-state complexity. Extending effective-memory calculations to such systems, and developing nonlinear causal coordinates that approach the local Hankel frontier over finite regions of record space, are natural next steps toward a general theory of reduced prediction in monitored quantum dynamics.

\section{Conclusion}
\label{sec:conclusion}

These results establish observable-specific memory as a natural basis for reduced quantum filtering. Within the finite classes, the exact closed structure provides the appropriate realization. Beyond that boundary, the relevant objective is a causal approximation that retains only the parts of the measurement record needed at the desired accuracy, rather than an exact closure of the conditional state. The observable Hankel framework makes this distinction precise and provides a lower frontier against which constructed filters can be compared. In this way, exact realizability and finite-accuracy predictive complexity become distinct, quantitatively testable properties of monitored nonlinear quantum dynamics.

\section*{Author Contributions}

The author conceived the project, developed the theoretical results, carried out the analytical derivations and numerical calculations, interpreted the results, and prepared the manuscript. OpenAI ChatGPT-5.6 Sol was used to assist with tightening prose, reviewing code, identifying possible errors in derivations and numerical analysis, and assisting with figure preparation. The author independently verified the calculations, numerical results, and final manuscript and takes full responsibility for the work.

\section*{Data Availability}

The code and numerical data used to generate the figures are available in the
Zenodo reproducibility package \cite{emerson_reproducibility_2026}.  The
archive includes fixed random seeds, generated measurement records,
record-level numerical outputs, refinement checks, and figure-generation
scripts.

\bibliographystyle{apsrev4-2}
\bibliography{references}
\clearpage
\twocolumngrid
\appendix
\section*{Supplemental Material}

\title{Supplemental Material for ``Finite Realizations and Effective Memory
in Monitored Nonlinear Quantum Dynamics''}

\author{Jacob Emerson}
\affiliation{Department of Electrical and Computer Engineering, Princeton University}

\date{\today}
\maketitle

This Supplemental Material proves the polynomial classification stated in the
main text and the record and measurement extensions used there.  It also gives
the numerical implementation and convergence details omitted from the main
text, together with additional robustness tests.

\section{Response matrices and physical records}
\label{sec:supp-response}

We use all notation from the main text.  For a constant controlled input
$u(t)=u_0$, write
\begin{equation}
 A={\cal H}_h+{\cal D}_{\bar n}
 -\frac12{\cal M}_\eta^2+u_0{\cal M}_\eta,
 \qquad
 c=\sqrt{2\eta\kappa}.
 \label{eq:supp-A}
\end{equation}
The measurement operator and the correction in this equation are
\begin{align}
 {\cal M}_\eta
 &=c\left(x+\frac12\partial_x\right),
 \label{eq:supp-M}\\
 -\frac12{\cal M}_\eta^2
 &=-\eta\kappa
 \left(x^2+x\partial_x+\frac12+\frac14\partial_x^2\right).
 \label{eq:supp-correction}
\end{align}
All calculations below apply finitely many polynomial-coefficient differential
operators to a Gaussian density and then integrate against a polynomial.  The
resulting expressions are therefore finite Gaussian moments.  Because the
Wigner evolution is linear and polynomial preserving in this setting, we may
track only the leading polynomial degrees when proving nonvanishing of the
selected minors.  Lower-degree terms can be discarded whenever they cannot
contribute to the leading signature being isolated.  When needed, we retain
specific lower-order Hamiltonian or measurement actions; these are selected
terms in the same differential evolution, retained because
they produce the required quadrature mixing for the chosen observable.

These operations are well defined, and integration by parts produces no
boundary terms.  Throughout this section the functions being paired are
Gaussian densities multiplied by polynomials, so no additional functional
analytic machinery is needed for the coefficient calculations.

We next make explicit why normalization and centering do not change the rank
test. Let \(Z(t)=\langle 1,W_t\rangle\), which is positive along the reference
trajectory. Dividing the centered unnormalized response by \(Z(t)\) amounts
to an invertible rescaling of the future output. By the product rule, differentiating the normalized response \(i\) times
either places all \(i\) derivatives on the centered unnormalized response,
giving it the nonzero coefficient \(Z(0)^{-1}\), or places at least one
derivative on \(Z^{-1}\). In the latter terms, fewer than \(i\) derivatives
act on the response. Normalization therefore cannot remove a new
independent response direction at order \(i\).

Differentiating the centering term \(\langle O\rangle_t\) produces tests of
the form \((A^\dagger)^q1\) with \(q<i\), rather than the leading test
\((A^\dagger)^i(o-\langle O\rangle_0)\). In the orderings used below, these
terms have lower future order and polynomial signature. Including the lower
rows allows them to be eliminated without changing the leading term of the
selected determinant.
It is enough to consider the reduced entries
\begin{equation}
 R_{ij}
 =
 \frac{\langle F_i,G_j\rangle}{\langle1,W_0\rangle},
 \qquad
 G_j=\operatorname{ad}_{-A}^{\,j}({\cal M}_\eta)W_0,
 \label{eq:supp-response-pairing}
\end{equation}
where
\[
 F_i=(A^\dagger)^i(o-\langle O\rangle_0),
\]
as in the main text.  We place the split at \(t=s=0\) and take \(W_0\) to be
the chosen Gaussian preparation there.  A nonzero derivative determinant at
this corner gives a nonzero evaluation determinant at sufficiently close,
distinct times.  These times may be chosen so that every past time precedes
every future time, as required by the separated intervals defining \(K_O\).

\begin{lemma}[Response-minor test]
\label{lem:supp-minor-test}
If the integral operator with kernel $K(t,s)$ has rank at most $d$, every
finite matrix of derivatives
\begin{equation}
 \left[
 \left.\partial_t^{i_r}\partial_s^{j_c}K(t,s)
 \right|_{t=s=0}
 \right]_{r,c=1}^N
 \label{eq:supp-derivative-matrix}
\end{equation}
has rank at most $d$.  Thus nonzero response minors of arbitrary size imply
$\rank K=\infty$.
\end{lemma}

\begin{proof}
A rank-$d$ smooth kernel has a local form
$K(t,s)=\sum_{a=1}^d f_a(t)g_a(s)$.  For chosen derivative orders
\(i_1,\ldots,i_N\) and \(j_1,\ldots,j_N\), define
\begin{equation}
 U_{ra}=f_a^{(i_r)}(0),
 \qquad
 V_{ac}=g_a^{(j_c)}(0).
\end{equation}
The derivative matrix then factors as
\begin{equation}
 \left[
 \left.\partial_t^{i_r}\partial_s^{j_c}K(t,s)
 \right|_{0,0}
 \right]_{r,c=1}^N
 =UV.
 \label{eq:supp-derivative-factorization}
\end{equation}
Because \(U\) has only \(d\) columns and \(V\) only \(d\) rows, this matrix
has rank at most \(d\), regardless of how many derivative orders are chosen.
Different derivatives may produce new functions, but after evaluation at the
split they remain entries of these same \(d\) future and past derivative
vectors. 
\end{proof}

The controlled constant-record calculation also applies to physical records.  The record law
has full support on continuous paths: every uniform neighborhood of a smooth
record has positive probability.  As used in the main text, robustness means that
the filter and the quantum record-to-output map extend continuously to these
controlled records and are differentiable along the smooth perturbations used
above.

\begin{proposition}[Passage to stochastic records]
\label{prop:supp-records}
If $K_O$ has infinite rank at a smooth controlled record, no
finite-dimensional robust $C^1$ filter is exact on a neighborhood of that
record or exact almost surely.  For every finite $d$, the inequality
$\rank K_O>d$ holds on an open set of physical records of positive
probability.
\end{proposition}

\begin{proof}
The derivative of a $d$-variable filter factors through its $d$-dimensional
state at the split and therefore has rank at most $d$.  It cannot equal an
infinite-rank derivative.  If two continuous record-to-output maps agreed
almost surely but differed at one record, continuity would give an open set
on which they differ, contradicting full support.  Finally, a nonzero
$(d+1)\times(d+1)$ response minor remains nonzero under a sufficiently small
change of the base record.  Its neighborhood has positive probability.
\end{proof}

\section{Proof of the one-mode classification}
\label{sec:supp-classification}

We first show how an increasing sequence of polynomial terms in the columns
$G_j$ produces response minors.  We then construct such a sequence for every
polynomial Hamiltonian outside ${\mathscr F}_X$.

\subsection{From polynomial columns to nonzero minors}

Use the displaced coherent Wigner function
\begin{equation}
 W_0(x,p)=\pi^{-1}e^{-(x-x_0)^2-(p-p_0)^2}.
 \label{eq:supp-coherent}
\end{equation}
Write
\begin{equation}
 B=x+\frac12\partial_x,
 \qquad
 C=x-\frac12\partial_x.
 \label{eq:supp-BC}
\end{equation}
Then ${\cal M}_\eta=cB$, $[B,C]=1$, and
\begin{equation}
 BW_0=x_0W_0,
 \qquad
 CW_0=(2x-x_0)W_0.
 \label{eq:supp-BC-action}
\end{equation}
The correction in Eq.~\eqref{eq:supp-correction} is $-\eta\kappa B^2$.
Since
\begin{equation}
 [B^2,C^r]=2rC^{r-1}B+r(r-1)C^{r-2},
 \label{eq:supp-B2C}
\end{equation}
it lowers the highest $C$ power and cannot cancel a term selected by that
power.

The coefficient of $u_0^i$ in the future test $F_i$ comes from choosing the
measurement part of $A$ at all $i$ steps.  It is therefore
$({\cal M}_\eta^\dagger)^i(o-\langle O\rangle_0)$; every other term has a
lower power of $u_0$.  Its generating action is
\begin{equation}
 e^{\lambda{\cal M}_\eta^\dagger}f(x,p)
 =e^{c\lambda x-c^2\lambda^2/4}
 f\left(x-\frac{c\lambda}{2},p\right).
 \label{eq:supp-measurement-flow}
\end{equation}
Since
\begin{equation}
 {\cal M}_\eta(x^ap^bW_0)
 =cx_0x^ap^bW_0+\frac{ca}{2}x^{a-1}p^bW_0,
 \label{eq:supp-measurement-lowering}
\end{equation}
the scalar part $cx_0$ produces only combinations of lower future rows.
After eliminating those rows, the leading coefficient is represented by
$({\cal M}_\eta^\dagger-cx_0)^i$.  For $o=x^mp^n$ and a column whose leading
term is $x^ap^bW_0$, this gives
\begin{align}
 &\left\langle
 ({\cal M}_\eta^\dagger-cx_0)^i
 (o-\langle O\rangle_0),x^ap^bW_0
 \right\rangle
 \notag\\
 &\quad=
 \left(\frac c2\right)^i\frac{a!}{(a-i)!}
 \notag\\
 &\qquad\times\Bigl[
 \langle P^{n+b}\rangle_0\langle X^{m+a-i}\rangle_0
 \notag\\
 &\qquad\quad
 -\langle X^m\rangle_0\langle P^n\rangle_0
  \langle P^b\rangle_0\langle X^{a-i}\rangle_0
 \Bigr]
 \label{eq:supp-common-pairing}
\end{align}
for $i\le a$, and the expression is zero for $i>a$.
All one-variable brackets in these Gaussian moment formulas denote
expectations in the preparation $W_0$.

\begin{lemma}[Gaussian extraction]
\label{lem:supp-gaussian-extraction}
Suppose an infinite subsequence of past columns has nonzero leading terms
$x^{a_r}p^{b_r}W_0$, where $a_r\to\infty$, and no lower term reaches the
displayed power at the same response order.
Then every observable $x^mp^n$ with $m\ge1$ gives nonzero response minors of
arbitrary size for generic $x_0,p_0,u_0$.  The same holds for $p^n$, $n\ge1$,
if $b_r>0$ along an infinite subsequence.
\end{lemma}

\begin{proof}
Choose a subsequence with $a_{r+1}-a_r>m$ and take the future derivative
order $i_r=a_r-m$.  Equation~\eqref{eq:supp-common-pairing} makes the leading
matrix triangular.  Its $r$th diagonal entry is a nonzero constant times
\begin{equation}
 \langle P^{n+b_r}\rangle_0\langle X^{2m}\rangle_0
 -\langle X^m\rangle_0^2
  \langle P^n\rangle_0\langle P^{b_r}\rangle_0.
 \label{eq:supp-x-pivot}
\end{equation}
As a polynomial in $p_0$, its highest coefficient is
$\langle X^{2m}\rangle_0-\langle X^m\rangle_0^2>0$.
A single generic $p_0$ avoids the roots of all countably many pivots.

For $m=0$, take $i_r=a_r$.  The diagonal becomes a nonzero constant times
\begin{equation}
 \langle P^{n+b_r}\rangle_0
 -\langle P^n\rangle_0\langle P^{b_r}\rangle_0.
 \label{eq:supp-p-pivot}
\end{equation}
For $b_r>0$, this is not the zero polynomial in $p_0$ because its leading
term is $(nb_r/2)p_0^{n+b_r-2}$.  Finally, each selected determinant is a
polynomial in $u_0$, and the triangular product above is its nonzero highest
coefficient.  One generic constant $u_0$ avoids the roots for every finite
minor.
\end{proof}

\subsection{Single nonlinear monomials}
\label{subsec:supp-monomials}

Consider first
\begin{equation}
 h(x,p)=\gamma x^\ell p^k,
 \qquad
 \gamma\ne0,
 \qquad
 \ell+k\ge3,
 \qquad
 k\ge1.
 \label{eq:supp-monomial}
\end{equation}
Its first-order Hamiltonian action is
\begin{equation}
 {\cal H}^{(1)}_{\ell k}
 =\gamma\left(
 \ell x^{\ell-1}p^k\partial_p
 -kx^\ell p^{k-1}\partial_x
 \right),
 \label{eq:supp-classical-field}
\end{equation}
and its first commutator with the measurement operator is
\begin{align}
 \frac{[{\cal H}^{(1)}_{\ell k},{\cal M}_\eta]}{c\gamma}
 ={}&-kx^\ell p^{k-1}
 +\frac{k\ell}{2}x^{\ell-1}p^{k-1}\partial_x
 \notag\\
 &-\frac{\ell(\ell-1)}2x^{\ell-2}p^k\partial_p.
 \label{eq:supp-first-response}
\end{align}

For $\ell\ge2$, begin with the multiplication term in
Eq.~\eqref{eq:supp-first-response}.  Since
\begin{equation}
 [{\cal H}^{(1)}_{\ell k},x^ap^b]
 =\gamma(\ell b-ka)x^{a+\ell-1}p^{b+k-1},
 \label{eq:supp-multiplication-commutator}
\end{equation}
The leading coefficient vanishes exactly when the exponent pairs
\((a,b)\) and \((\ell,k)\) are proportional.  The construction therefore
does not require consecutive response derivatives: when a selected exponent
pair is proportional to \((\ell,k)\), we pass to the first past direction or
future test whose leading exponent is not proportional to it.

There is a second, independent reason to select derivative orders.  For the
centered coherent preparation, a Gaussian moment of \(x^rp^q\) vanishes when
either \(r\) or \(q\) is odd.  Such a parity zero does not terminate the
degree-raising sequence.  The controlled generator contains
\(u_0{\cal M}_\eta\), and the leading action of
\({\cal M}_\eta^\dagger\) adds one power of the measured coordinate.  With
\(u_0\ne0\), the derivative sequences therefore contain choices with the
opposite \(x\) parity.  Equivalently, a generic displacement removes the
centered parity symmetry.  The proof below chooses whichever row or column
has a nonvanishing Gaussian pairing.

The selected powers after $j$ Hamiltonian steps are
\begin{equation}
 a_j=1+j(\ell-1),
 \qquad
 b_j=j(k-1),
 \label{eq:supp-multiplication-powers}
\end{equation}
The first coefficient is $-c\gamma k$, and the $j$th Hamiltonian step
multiplies it by $\gamma[j(k-\ell)-k]$.  The sequence therefore continues
forever unless
\begin{equation}
 k>\ell,
 \qquad
 \frac{k}{k-\ell}\in\mathbb N.
 \label{eq:supp-resonance}
\end{equation}

When Eq.~\eqref{eq:supp-resonance} holds, use instead the derivative terms in
Eq.~\eqref{eq:supp-first-response}.  Their leading seed is
\begin{equation}
 x^{\ell-2}p^{k-1}
 \left(kx\partial_x-(\ell-1)p\partial_p\right).
\end{equation}
If the two derivative coefficients are placed in a vector $\bm v_r$, the
next Hamiltonian commutator gives
\begin{equation}
 \bm v_{r+1}
 =
 \begin{pmatrix}
  k+(k-1)\ell+r(k-\ell)&k(k-1)\\
  -\ell(\ell-1)&k(2-\ell)+r(k-\ell)
 \end{pmatrix}\bm v_r,
 \label{eq:supp-derivative-recurrence}
\end{equation}
starting from $\bm v_0=(k,-(\ell-1))^{\mathsf T}$.  The determinant of the
matrix in Eq.~\eqref{eq:supp-derivative-recurrence} is
\begin{equation}
 [2k-\ell+r(k-\ell)][2k-\ell+(r-1)(k-\ell)].
 \label{eq:supp-derivative-determinant}
\end{equation}
On the exceptional set $k>\ell$, both factors in the last line are positive.
Thus the derivative sequence is nonzero at every order precisely when the
multiplication sequence can stop.  Acting on $W_0$ changes its derivatives
into polynomials times $W_0$, with unbounded $x$ degree.

The displayed terms remain leading terms of the exact Moyal evolution.  A
higher Moyal term differentiates the Hamiltonian at least three times and
therefore loses at least two polynomial powers relative to
Eq.~\eqref{eq:supp-classical-field}.  Damping, diffusion, and measurement are
lower when terms are ordered first by momentum degree and then by $x$ degree.
On the $k=1$ edge, Eq.~\eqref{eq:supp-B2C} gives the needed ordering instead.
The sequence therefore satisfies Lemma~\ref{lem:supp-gaussian-extraction}.

There is one extra step when $k=1$ and the observable is $P^n$.  The past
columns then have no increasing $p$ power.  Insert one Hamiltonian action in
the future.  Since
\begin{equation}
 {\cal H}^{(1)\dagger}_{\ell 1}(x^ap^n)
 =\gamma(a-n\ell)x^{a+\ell-1}p^n,
\end{equation}
summing its $r+1$ possible positions among $r$ measurement actions leaves
\begin{equation}
 -(r+1)n\ell\gamma(cu_0)^r x^{r+\ell-1}p^n.
 \label{eq:supp-xlP-row}
\end{equation}
After the same lower-row elimination and removal of the nonzero generic
$p$-moment factors, matching its \(x\) degree to the past columns gives the
Gaussian moment matrix
\begin{equation}
 \left[\langle X^{a_r+a_s}\rangle_0\right]_{r,s=1}^N.
 \label{eq:supp-x-gauss}
\end{equation}
This is the Gram matrix of the distinct monomials \(X^{a_r}\) in
\(L^2(W_0)\).  They are linearly independent because a nonzero polynomial
cannot vanish almost surely under a nondegenerate Gaussian distribution.
Hence the matrix is positive definite.

For a pure momentum Hamiltonian $h=\gamma p^k$, $k\ge3$, momentum diffusion
supplies the missing step.  Its coefficient
$\kappa(2\bar n+1)/4$ is strictly positive.
For $Q_j=p^{a_j}\partial_x^j$, the leading double commutator is
\begin{align}
 &\left[{\cal H}^{(1)}_{0k},
 \left[\frac{\kappa(2\bar n+1)}4\partial_p^2,Q_j\right]\right]
 \notag\\
 &\quad=\frac{\kappa(2\bar n+1)}2\gamma k(k-1)a_j
 p^{a_j+k-3}\partial_x^{j+1}
 +\text{lower terms}.
 \label{eq:supp-pure-p-cycle}
\end{align}
Starting with $a_0=k-1$ gives
\begin{equation}
 Q_j=p^{k-1+j(k-3)}\partial_x^j,
 \qquad j=0,1,\ldots,
 \label{eq:supp-pure-p-sequence}
\end{equation}
with a nonzero coefficient at every step.  Since
$\partial_x^jW_0=(-2)^jx^jW_0$ plus lower $x$ powers, both requirements of
Lemma~\ref{lem:supp-gaussian-extraction} hold for every observable.

It remains to treat $h=\gamma xp^k$, $k\ge2$.  Equation
\eqref{eq:supp-first-response} contains $p^{k-1}\partial_x$, and
\begin{equation}
 [{\cal H}^{(1)}_{1k},p^b\partial_x]
 =\gamma(b+k)p^{b+k-1}\partial_x.
\end{equation}
The past columns therefore contain
\begin{equation}
 p^{j(k-1)}\partial_xW_0,
 \qquad j=1,2,\ldots,
 \label{eq:supp-xpk-columns}
\end{equation}
with nonzero coefficients.  These columns grow in $p$ rather than $x$, so we
choose future rows with both measurement and Hamiltonian actions.  Let
$S_{R,q}(m,n)$ be the leading coefficient after $R$ measurement actions and
$q$ Hamiltonian actions, with the common factor $c^R\gamma^q$ removed.
Sorting the words by their last action gives
\begin{align}
 S_{R,q}={}&S_{R-1,q}
 +[k(m+R)-n-(q-1)(k-1)]S_{R,q-1},
 \notag\\
 &S_{R,0}=1.
 \label{eq:supp-word-recurrence}
\end{align}
Induction and discrete summation give
\begin{equation}
 S_{R,q}=\frac{k^q}{2^q q!}R^{2q}+O(R^{2q-1}).
 \label{eq:supp-word-leading}
\end{equation}
Thus, for every $q$, arbitrarily large $R$ give a nonzero row whose leading
monomial is $x^{m+R}p^{n+q(k-1)}$.

After pairing these rows with Eq.~\eqref{eq:supp-xpk-columns}, the remaining
matrix has entries
\begin{equation}
 \langle P^{a_r+b_s}\rangle_0,
 \qquad
 a_1<\cdots<a_N,
 \qquad
 b_1<\cdots<b_N.
 \label{eq:supp-general-moment-matrix}
\end{equation}
Its determinant is not the zero polynomial in $p_0$.  To see this, write
$P=p_0+Z$ and write the determinant as the integral of the two determinants
formed from the powers $a_r$ and $b_s$.  In the highest nonzero power of
$p_0$, each of these two determinants reduces to a nonzero constant times the
Vandermonde determinant in the integration variables.  Their product is
therefore a positive constant times
\begin{equation}
 \int_{\mathbb R^N}
 \prod_{i<j}(z_j-z_i)^2
 \prod_{i=1}^N e^{-z_i^2}\,dz_i>0.
 \label{eq:supp-vandermonde}
\end{equation}
Hence a generic displacement makes every required determinant nonzero.  This
completes all single-monomial cases.

\subsection{General polynomial Hamiltonians}
\label{subsec:supp-general-polynomial}

Write the polynomial as
\begin{equation}
 h(x,p)=\sum_{k=0}^{k_{\max}} f_k(x)p^k.
 \label{eq:supp-collect-p}
\end{equation}
For a differential monomial
$x^ap^b\partial_x^r\partial_p^s$, order its terms by
\begin{equation}
 (b-s,a-r)
 \label{eq:supp-weight}
\end{equation}
in lexicographic order.  The first-order Hamiltonian term produced by
$x^\ell p^k$ has order $(k-1,\ell-1)$, while every higher Moyal term is
strictly lower.  Pure $x$ potentials, damping, diffusion, and measurement are
also lower than a selected term with positive momentum order.  For $k=1$,
order instead by momentum degree and then by the $C$ power in
Eq.~\eqref{eq:supp-BC}; Eq.~\eqref{eq:supp-B2C} again places the correction
below the selected term.

It follows that the single-monomial calculation survives unchanged if the
Hamiltonian has

\begin{itemize}
\item $k_{\max}\ge3$;
\item $k_{\max}=2$ and $f_2$ is nonconstant; or
\item no $p^2$ term and $\deg f_1\ge2$.
\end{itemize}

Indeed, in each line the highest $x$ power in the stated coefficient is the
unique largest nonlinear term in the relevant ordering.  Replacing one copy
of it by any other Hamiltonian term lowers the final order, so no collection
of lower terms can cancel the leading response sequence.

If none of these cases applies and the Hamiltonian is outside
${\mathscr F}_X$, it must have the form
\begin{equation}
 h(x,p)=\frac\beta2p^2+F(x)p+V(x),
 \qquad \beta\ne0.
 \label{eq:supp-shear-family}
\end{equation}
Set
\begin{equation}
 \begin{aligned}
 s(x)&=\frac{F(x)}\beta,
 &q&=p+s(x),\\
 U(x)&=V(x)-\frac{F(x)^2}{2\beta}.
 \end{aligned}
 \label{eq:supp-shear-change}
\end{equation}
Then
\begin{equation}
 h(x,p)=\frac\beta2q^2+U(x),
 \qquad \{x,q\}=1.
 \label{eq:supp-shear-form}
\end{equation}
This change leaves the measured coordinate $x$ fixed.  When $s$ is
nonlinear, we use it only to rewrite the differential generator; we do not
assume that the Moyal product is unchanged by the transformation.

If $s$ is affine, it preserves polynomial degree.  When $U$ is at most
quadratic the original Hamiltonian is Gaussian.  Otherwise, if
$U(x)=u_dx^d+\cdots$ with $d\ge3$, the two leading actions are
\begin{equation}
 x^ap^b\longmapsto x^{a-1}p^{b+1},
 \qquad
 x^ap^b\longmapsto x^{a+d-1}p^{b-1}.
\end{equation}
Alternating them gives
\begin{equation}
 x^{1+r(d-2)}\longmapsto x^{r(d-2)}p
 \longmapsto x^{1+(r+1)(d-2)},
 \label{eq:supp-kinetic-potential}
\end{equation}
with nonzero coefficients.  Higher Moyal terms lose at least two powers of
$x$, so they cannot cancel this sequence.  Lemma
\ref{lem:supp-gaussian-extraction} treats observables with positive $x$
degree.  For $P^n$, one potential action in the future gives a positive
Gaussian moment matrix of the form Eq.~\eqref{eq:supp-x-gauss}.

Suppose now that $s$ is nonlinear.  At fixed $q$,
\begin{equation}
 \partial_x\big|_p=\partial_x\big|_q+s'(x)\partial_q,
 \qquad
 \partial_p=\partial_q.
 \label{eq:supp-shear-derivatives}
\end{equation}
The $x$-dependent coefficient of $\partial_q$ in the full controlled
generator is
\begin{align}
 f_{u_0}(x)={}&U'(x)
 +\frac\kappa2[xs'(x)-s(x)]
 \notag\\
 &-\eta\kappa\left[xs'(x)+\frac14s''(x)\right]
 +\frac{cu_0}{2}s'(x).
 \label{eq:supp-shear-force}
\end{align}
The two terms relevant at highest degree are then
\begin{equation}
 -\beta q\partial_x
 \qquad\text{and}\qquad
 f_{u_0}(x)\partial_q.
 \label{eq:supp-shear-leading-actions}
\end{equation}
If $d=\deg f_{u_0}\ge2$, write its leading term as
$\lambda x^d$, with $\lambda\ne0$.  At leading order, the two actions are
\begin{equation}
 x^aq^b\longmapsto-\beta a x^{a-1}q^{b+1},
 \qquad
 x^aq^b\longmapsto\lambda b x^{a+d}q^{b-1}.
 \label{eq:supp-shear-degree-actions}
\end{equation}
The first raises the $q$ degree by one and lowers the $x$ degree by one; the
second reverses that change and adds $d$ powers of $x$.  A word that begins with a pure $x$ term,
returns to zero $q$ degree after $2r$ steps, and reaches the largest possible
$x$ degree must contain $r$ actions of each type and can never apply the
force at zero $q$ degree.  Every such word has sign
$(-\beta\lambda)^r$, multiplied by positive integer degree factors.  The
alternating word is one such nonzero word, so the terms at the largest degree
cannot cancel.  The resulting columns contain nonzero leading terms
$x^{1+r(d-1)}W_0$ for arbitrarily large $r$.

If $\deg s\ge3$, a generic $u_0$ makes
$\deg f_{u_0}\ge2$.  If the record-independent nonlinear part in
Eq.~\eqref{eq:supp-shear-force} vanishes, the term $u_0s'$ supplies it.
Only
\begin{equation}
 s(x)=\alpha x^2+ax+b,
 \qquad \alpha\ne0,
 \label{eq:supp-quadratic-shear}
\end{equation}
can avoid this step.  It does so only when $U$ is tuned so that
\begin{equation}
 U'''=\kappa\alpha(4\eta-1).
 \label{eq:supp-tuned-U}
\end{equation}

We finish this last case directly.  Put
\begin{equation}
 d_x=\frac{\kappa(2\bar n+1-\eta)}4.
 \label{eq:supp-diffusions}
\end{equation}
At the differential orders needed below, the relevant parts of the generator
are
\begin{align}
 &-\beta q\,\partial_x\big|_q+U'(x)\partial_p,
 \notag\\
 &d_x(\partial_x\big|_p)^2
 +\frac{\kappa(2\bar n+1)}4\partial_p^2,
 \notag\\
 &\frac{\beta\alpha}{4}\partial_p^2\partial_x\big|_q
 +\text{constant}\times\partial_p^3.
 \label{eq:supp-tuned-operators}
\end{align}
Commuting the last line with the derivative part of the measurement operator
produces the nonzero starting term
\begin{equation}
 G_1=-\frac{c\beta\alpha^2}{4}\partial_p^3W_0
 +\text{lower terms}.
 \label{eq:supp-tuned-start}
\end{equation}

When $d_x>0$, the unique highest three-step cycle is
\begin{equation}
 \partial_p^n
 \longmapsto(\partial_x\big|_q)\partial_p^{n-1}
 \longmapsto(\partial_x\big|_p)\partial_p^n
 \longmapsto\partial_p^{n+3}.
 \label{eq:supp-tuned-cycle}
\end{equation}
The three arrows come from the kinetic term, the transformed $x$ diffusion,
and the third-order Moyal term, respectively.
Its coefficient is $2nd_x\beta^2\alpha^3\ne0$.  Repetition therefore gives
$G_{1+3r}$ a nonzero leading term $\partial_p^{3(r+1)}W_0$.
When $d_x=0$, necessarily $(\bar n,\eta)=(0,1)$, and
Eq.~\eqref{eq:supp-tuned-U} gives $U'''=3\kappa\alpha\ne0$.  The replacement
cycle
\begin{equation}
 \partial_p^n
 \longmapsto(\partial_x\big|_q)\partial_p^{n-1}
 \longmapsto x\partial_p^n
 \longmapsto\partial_p^{n+2}
\label{eq:supp-tuned-ideal-cycle}
\end{equation}
Here the middle arrow comes from the cubic part of $U$, and the final arrow
again comes from the third-order Moyal term.
The cycle has coefficient $n\beta^2\alpha U'''/4\ne0$.  Hence in this case
$G_{1+3r}$ has a nonzero leading term $\partial_p^{3+2r}W_0$.

Let $M_r$ denote the $p$-derivative order in the $r$th selected column.
Every observable sees either sequence.  For $x^mp^n$ with $n\ge1$, choose
future kinetic actions to raise the $p$ degree to $M_r$ and measurement
actions to supply the required $x$ degree.  The associated falling-factorial
coefficient is nonzero.  Integration by parts then gives a
diagonal term
\begin{equation}
 (-1)^{M_r}M_r!\langle X^m\rangle_0,
 \label{eq:supp-tuned-p-pairing}
\end{equation}
which is nonzero for generic $x_0$.  For $n=0$ and $m\ge1$, the centered
readout is treated in the same way, and the diagonal pairing is
\begin{equation}
 (-1)^{M_r}M_r!
 \left(\langle X^{2m}\rangle_0-\langle X^m\rangle_0^2\right)\ne0.
 \label{eq:supp-tuned-x-pairing}
\end{equation}
Ordering the columns by increasing $M_r$ makes the matrix triangular, since
a polynomial of $p$ degree $M_r$ pairs to zero with
$\partial_p^{M_s}W_0$ when $M_s>M_r$.

We have now covered every polynomial outside ${\mathscr F}_X$: either a
largest nonlinear monomial produces one of the sequences in
Sec.~\ref{subsec:supp-monomials}, or completing the square gives one of the
three shear cases above.  Lemma~\ref{lem:supp-minor-test} then gives infinite
response rank for the coherent preparation.

\subsection{Almost every Gaussian preparation}

A one-mode Gaussian preparation is determined by its displacement $\bm\mu$
and positive covariance matrix $\Sigma$, with
$\det\Sigma\ge1/4$.  Fix one of the response determinants constructed above.
Each of its entries is a finite Gaussian moment after finitely many
polynomial differentiations.  Its determinant is therefore a real-analytic
function of $\bm\mu$ and $\Sigma$.

For each size $N$, that function is not identically zero because the coherent
calculation above supplies a nonzero value.  Its zero set consequently has
measure zero.  The union of these zero sets over
$N=1,2,\ldots$ still has measure zero.  Outside that union, nonzero minors of
every size exist, so the response rank is infinite.  Restricting the same
determinants to the displacement and squeezing parameters gives the same
argument for pure Gaussian states, because the coherent witnesses are pure.
This proves the ``almost every Gaussian'' statement in the main theorem.
Together with Proposition~\ref{prop:supp-records}, it also proves the claimed
no-go result for robust stochastic filters.

\section{Measurement and network extensions}
\label{sec:supp-extensions}

Only the steps not already contained in the one-mode proof are needed here.
A phase-space rotation sends $(X,P)$ to $(X_\theta,P_\theta)$, preserves the
Weyl ordering and Gaussian preparations, and sends the measurement operator to
the same expression in $X_\theta$.  Applying the proof in the rotated
coordinates gives the class ${\mathscr F}_\theta$ stated in the main text.

For heterodyne monitoring, vary one real record channel while holding the
other fixed.  If the two-input response operator had finite rank, its
restriction to this one-channel input would also have finite rank.  A
nonlinear Hamiltonian cannot be in the aligned CMM class for both conjugate
quadratures.  At least one restriction is therefore covered by the one-mode
proof.  The fixed channel contributes only Gaussian measurement and
correction terms, which are lower in the orderings used above and cannot
cancel the selected leading terms.  This leaves only the Gaussian class as
the finite heterodyne class.

The proof also covers every $0<\eta\le1$ and $\bar n\ge0$.  The record
operator is only rescaled by $\sqrt\eta$, while the momentum-diffusion
coefficient $\kappa(2\bar n+1)/4$ is always positive.  The only point where
$d_x$ vanishes is $(\bar n,\eta)=(0,1)$, which is the second tuned-shear cycle
in Eq.~\eqref{eq:supp-tuned-ideal-cycle}.

Finally, suppose an obstructed mode is coupled polynomially to a finite
network with coupling parameters $\bm\lambda$.  At
$\bm\lambda=0$, every selected response determinant reduces to the nonzero
one-mode determinant.  For fixed $N$, the determinant is analytic in
$\bm\lambda$ and thus remains nonzero near zero and away from a measure-zero
set.  Taking the countable intersection over $N$ gives the generic
persistence statement in the main text.  This argument does not classify
nonlinearities involving only the other modes, and it does not exclude
exceptional finely tuned coupling values.

\section{Numerical implementation and robustness tests}
\label{sec:supp-numerics}

\subsection{Response kernels and spectral calculations}

All calculations set \(\kappa=1\) and measure time in units of
\(1/\kappa\).  The homodyne state trajectories and their tangent equations
are propagated with the same second-order split-step update.  The tangent
update is obtained by differentiating the discrete state update, so the
sampled response matrix is the derivative of the numerical record-to-output
map rather than a separately discretized continuum expression.  A
perturbation at the past-grid point \(s_j\) initializes
\begin{equation}
 \delta\widetilde\rho^{(j)}(s_j^+)
 ={\cal M}_\eta\widetilde\rho(s_j),
\end{equation}
and propagating this state gives
\begin{equation}
 K_{ij}
 =
 \frac{
 \left\langle
 o-\langle O\rangle_{t_i},\delta W^{(j)}(t_i)
 \right\rangle}
 {\langle1,W(t_i)\rangle}.
 \label{eq:supp-sampled-kernel}
\end{equation}

To approximate the continuous map
\(K_O:L^2(I_-)\rightarrow L^2(I_+)\), the sampled kernel is weighted before
taking its matrix SVD:
\begin{equation}
 \widetilde K_{ij}
 =
 \sqrt{w_i^+}\,K_{ij}\sqrt{w_j^-},
 \label{eq:supp-weighted-kernel}
\end{equation}
where \(w_i^+\) and \(w_j^-\) are the quadrature weights on the future and
past grids.  Without these factors, the singular values would depend on the
grid normalization rather than approximate those of the continuous
\(L^2\) operator.

We validate the tangent construction by adding
\(dY_t\mapsto dY_t\pm\varepsilon h(t)dt\) to the same physical record and
comparing \(K_Oh\) with the central difference
\begin{equation}
 D_\varepsilon h
 =
 \frac{
 \langle O\rangle[Y+\varepsilon h]
 -\langle O\rangle[Y-\varepsilon h]}
 {2\varepsilon}.
 \label{eq:supp-central-difference}
\end{equation}
The discrepancy has the expected \(O(\varepsilon^2)\) behavior and reaches
\(2.74\times10^{-9}\) at \(\varepsilon=10^{-4}\) in the test shown in the
main text.

For the spectral comparison, the past and future windows both have length
\(0.55\), separated by a gap of \(0.06\), with time step \(0.0025\) and Fock
cutoff \(30\).  The common displacement is
\(\alpha=1.4+0.55i\).  The Gaussian control uses \(H=0.85N\), \(O=X\), and
\(D(\alpha)S(0.6)\lvert0\rangle\); the CMM control uses
\(H=0.10X^3\) and \(O=P\).  The nonlinear curves use
\(K/\kappa=3.2\) for Kerr and \(g/\kappa=0.22\) for Duffing, with \(O=X\)
and coherent initial states.  The Kerr ensemble contains 24 independently
generated physical homodyne records.  Time-step and cutoff refinement
separate the nonterminating Kerr and Duffing tails from the second-order
discretization floor of the rank-two CMM control.

\subsection{Finite-amplitude response models}

For each physical record, the retained dimension is
\(d=d_{10^{-3}}\).  The quadratic and cubic decoders are fitted at radii
\(0.05\), \(0.15\), \(0.3\), and \(0.6\) using both signs of every
perturbation.  Kerr uses 48 independently drawn training directions and
Duffing uses 56.  The linear coefficient is fixed by the response SVD, while
the quadratic and combined quadratic--cubic coefficients are obtained by
ridge least squares.  All reported errors use independently drawn test
directions.  These directions lie in the space spanned by the first ten
singular functions, while the decoder receives only its first \(d\)
coordinates; the test therefore includes the effect of omitted response
directions whenever \(d<10\).

The Kerr test uses \(K/\kappa=0.8\), \(O=X\), 24 physical records, coherent
amplitude \(2\), Fock cutoff \(24\), and time step \(0.01\).  The past, gap,
and future lengths are \(0.8\), \(0.08\), and \(0.8\).  The Duffing test uses
\(g/\kappa=0.22\), 20 records, coherent amplitude \(1.4+0.55i\), and Fock
cutoff \(30\), with the same time grid and window lengths.  The omitted-mode
test uses radius \(r=0.03\) and perturbs each model along its first omitted
singular function.

\subsection{Poisson--Charlier and Fock calculations}

The Poisson--Charlier coefficient equations and the direct Fock amplitudes
are driven by identical complex heterodyne records and advanced with the same
normalized Milstein discretization.  Model-dependent matrices are
precomputed and excluded from the reported online propagation times.  In the
fixed-\(N_0=9\) convergence test, \(K/\kappa=0.4\),
\(\kappa T=0.8\), the time step is \(0.002\), and the Fock reference cutoff
is \(94\).  The Charlier weight uses
\(\mu=N_0e^{-\kappa T}\).  Forty physical records are used.  Repeating the
order-32 calculation at time step \(0.001\) gives a 90th-percentile endpoint
state error of \(3.63\times10^{-8}\), below the plotted-grid result.

The response-dimension comparison uses \(N_0=9\),
\(K/\kappa=0.4\), \(T=0.75\), time step \(0.005\), and a 43-component
complex Fock reference.  The two real heterodyne response blocks are
concatenated before taking the weighted SVD.  The Charlier response matrix is
obtained by differentiating its truncated coefficient evolution with respect
to the same two record channels.

For the occupation scan, \(T=0.5\) and the parameters are chosen so that
\begin{equation}
 \frac{KN_0(1-e^{-\kappa T})}{\kappa}=1.
\end{equation}
At each occupation, the reported state dimension is the smallest tested
value whose 90th-percentile phase-aligned endpoint error over 20 physical
records is at most \(10^{-6}\).  Timing values are medians of five blocks
with five repeats per block on the same machine.  The two methods process
identical batches of 128 records, and reusable operator construction is
excluded.  These timings compare the implementations under matched
conditions and are not intended as hardware-independent absolute costs.

The response spectra were checked under time-step and Fock-cutoff refinement.
Pure-state and density-matrix response kernels agree to relative error below
\(10^{-8}\).  Full parameters, fixed seeds, record-level results, refinement
data, timing measurements, and figure scripts are included in the deposited
reproducibility archive cited in the main text.

\subsection{Additional robustness tests}

Figure~\ref{fig:supp-robustness} adds four checks not displayed in the main
text: exact CMM time-step convergence, dependence on the chosen future
observable, inefficient monitoring, and time-step convergence of the Kerr
spectrum.

\begin{figure*}[t]
 \centering
 \includegraphics[width=\textwidth]{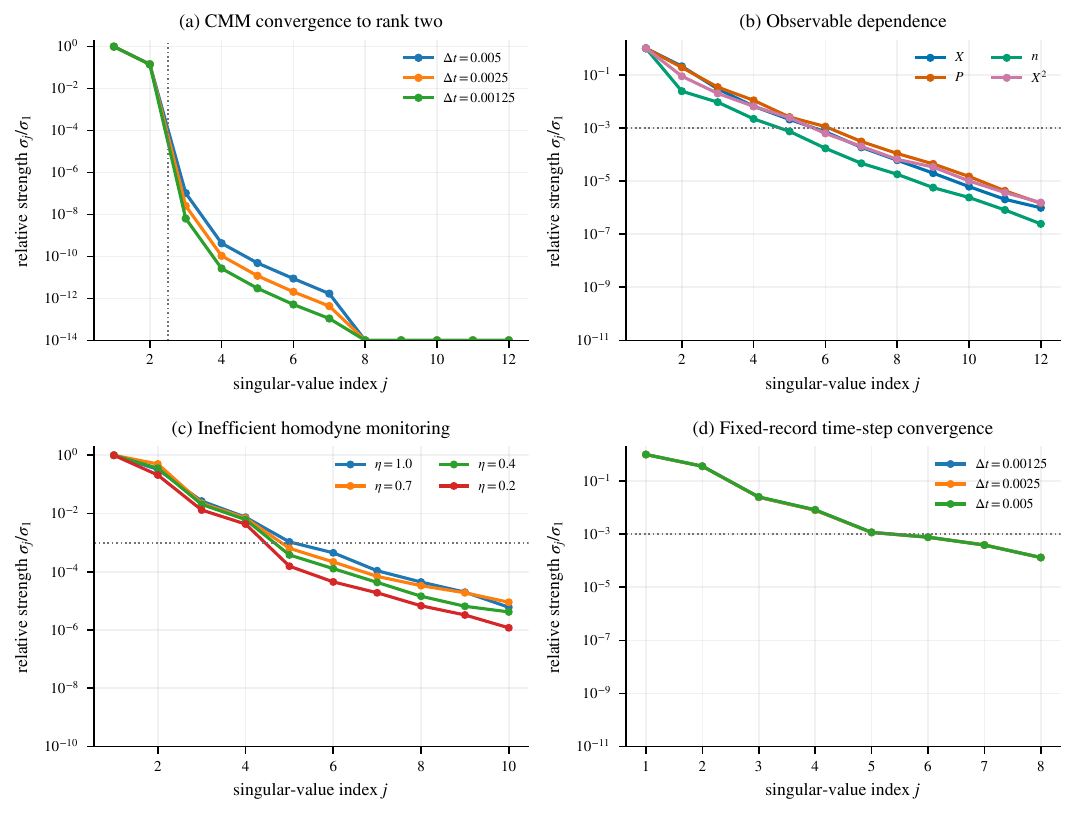}
 \caption{Additional checks of the observable-response spectrum.
 (a) For the cubic CMM control, the first two singular values remain finite
 while the apparent third decreases by approximately four under each
 time-step halving, as expected for an exact rank-two response computed with
 a second-order step.  (b) Median Kerr spectra for the future observables
 $X$, $P$, $N$, and $X^2$ over 16 shared physical records.  The decay depends
 on the observable because $K_O$ retains only the state directions visible to
 that output.  (c) Median Kerr spectra from 12 density-matrix records at each
 efficiency.  The low-dimensional decay remains visible from $\eta=1$ to
 $\eta=0.2$.  (d) A fixed-record Kerr spectrum under successive time-step
 halvings.  The nonzero tail converges rather than retreating to the numerical
 floor.  Every panel uses separated past and future windows and the same
 quadrature weighting as the main figures.}
 \label{fig:supp-robustness}
\end{figure*}

For panel (c), the median value of $d_{10^{-3}}$ stays between four and
$4.5$ across the displayed efficiencies.  At $\eta=0.4$, halving the time
step leaves $d_{10^{-3}}=4$ and changes the first three nontrivial normalized
singular values by $6.3\times10^{-4}$, $1.1\times10^{-4}$, and
$2.9\times10^{-5}$.  Together with panels (a) and (d), this separates a true
finite-rank termination from a converged but decaying infinite-rank tail.

\end{document}